\documentclass[sigconf,screen,authorversion,nonacm]{acmart}

\AtBeginDocument{%
  \providecommand\BibTeX{{%
    \normalfont B\kern-0.5em{\scshape i\kern-0.25em b}\kern-0.8em\TeX}}}

\usepackage[svgnames]{xcolor}

\usepackage[english]{babel}
\usepackage{rotating}
\usepackage{csquotes}
\usepackage{multirow}
\usepackage{siunitx}
\usepackage{balance}

\usepackage{graphicx}
\usepackage{booktabs}

\usepackage[labelformat=parens]{subcaption}

\usepackage{mathtools}
\usepackage{amsthm}
\newtheorem{definition}{Definition}
\newtheorem*{definition*}{Definition}
\newtheorem{lemma}{Lemma}
\newtheorem{theorem}{Theorem}
\newtheorem{claim}{Claim}
\newtheorem*{claim*}{Claim}

\newtheorem*{corollary*}{Corollary}

\theoremstyle{definition}
\usepackage{cancel}
\usepackage{stfloats}

\usepackage[noend]{algpseudocode}
\usepackage{algorithm}

\usepackage[textsize=small]{todonotes}

\newcommand{\defeq}{\vcentcolon=}
\newcommand{\bigosym}{\mathcal{O}}
\newcommand{\vol}[1]{\operatorname{vol}(#1)}
\newcommand{\norm}[1]{\left\lVert#1\right\rVert}

\DeclareMathOperator{\EX}{\mathbb{E}}%
\DeclareRobustCommand{\bbone}{\text{\usefont{U}{bbold}{m}{n}1}} %
\newcommand{\nc}{\ensuremath{\theta}}
\newcommand{\cutParent}{\operatorname{cutParent}}

\newcommand{\MaxDegree}{\ensuremath{\Delta}}
\newcommand{\Algo}[1]{$\textsf{XCut}_\textsf{#1}$}
\newcommand{\Metis}{\textsf{METIS}}
\newcommand{\Kahip}{\textsf{KaHiP}}
\newcommand{\Graclus}{\textsf{Graclus}}

\newcommand{\ea}{et al.}

\newcommand{\CitNet}{\textsf{CN}}
\newcommand{\SocNet}{\textsf{SN}}
\newcommand{\TriMix}{\textsf{TM}}
\newcommand{\Road}{\textsf{RN}}
\newcommand{\CircSim}{\textsf{CS}}
\newcommand{\Web}{\textsf{WB}}
\newcommand{\Infra}{\textsf{IF}}
\newcommand{\Opti}{\textsf{OP}}
\newcommand{\Email}{\textsf{EM}}
\newcommand{\NumSim}{\textsf{NS}}
\newcommand{\Clus}{\textsf{CL}}
\newcommand{\FinEl}{\textsf{FE}}
\newcommand{\ComFlu}{\textsf{CF}}
\newcommand{\DupMat}{\textsf{DM}}
\newcommand{\Bip}{\textsf{BP}}
\newcommand{\RandG}{\textsf{RD}}
\newcommand{\USCens}{\textsf{US}}

\usepackage{wrapfig}
\newcommand{\erclogowrapped}[1]{%
\setlength\intextsep{0pt}%
\begin{wrapfigure}[3]{r}{#1*\real{1.1}}%
\includegraphics[width=#1]{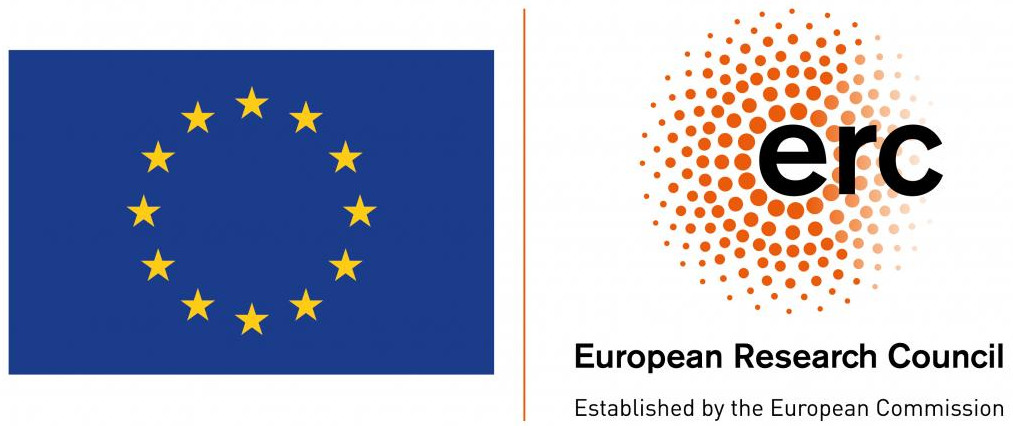}%
\end{wrapfigure}%
}

\copyrightyear{2024}
\acmYear{2024}
\setcopyright{cc}
\acmConference[KDD '24]{Proceedings of the 30th ACM SIGKDD Conference on 
Knowledge Discovery and Data Mining}{August 25--29, 2024}{Barcelona, Spain}
\acmBooktitle{Proceedings of the 30th ACM SIGKDD Conference on Knowledge 
Discovery and Data Mining (KDD '24), August 25--29, 2024, Barcelona, 
Spain}
\acmDOI{10.1145/3637528.3671978}
\acmISBN{979-8-4007-0490-1/24/08}

\begin{document}

\title{Expander Hierarchies for Normalized Cuts on Graphs}
\thanks{Full version of:
Kathrin Hanauer, Monika Henzinger, Robin Münk, Harald Räcke, and Maximilian Vötsch. 2024. Expander Hierarchies for Normalized Cuts on Graphs.  In \emph{Proceedings of the 30th ACM SIGKDD Conference on Knowledge Discovery and Data Mining (KDD ’24), August 25–29, 2024, Barcelona, Spain. ACM, New York, NY, USA}, 12 pages. \url{https://doi.org/10.1145/3637528.3671978}~\cite{kddofficial}}

\author{Kathrin Hanauer}
\email{kathrin.hanauer@univie.ac.at}
\orcid{0000-0002-5945-837X}
\affiliation{%
\institution{University of Vienna}
\department{Faculty of Computer Science}
\city{Vienna}
\country{Austria}
}

\author{Monika Henzinger}
\email{monika.henzinger@ist.ac.at}
\orcid{0000-0002-5008-6530}
\affiliation{%
\institution{Institute of Science and Technology Austria (ISTA)}
\city{Klosterneuburg}
\country{Austria}
}

\author{Robin Münk}
\email{robin.muenk@tum.de}
\orcid{0009-0000-5083-7704} %
\affiliation{%
\institution{Technical University of Munich}
\city{Munich}
\country{Germany}
}
\author{Harald Räcke}
\email{raecke@in.tum.de}
\orcid{0000-0001-8797-717X}
\affiliation{%
\institution{Technical University of Munich}
\city{Munich}
\country{Germany}
}

\author{Maximilian Vötsch}
\email{maximilian.voetsch@univie.ac.at}
\orcid{0009-0006-6793-6745}
\affiliation{%
\institution{University of Vienna}
\department{Faculty of Computer Science}
\department{UniVie Doctoral School Computer~Science~DoCS}
\city{Vienna}
\country{Austria}
}

\begin{abstract}
Expander decompositions of graphs have significantly advanced the understanding of many classical graph problems and led to numerous fundamental theoretical results.
However, their adoption in practice has been hindered
due to their inherent intricacies and large hidden factors in their asymptotic running times.
Here, we introduce %
the first
practically efficient algorithm for computing expander decompositions and their hierarchies
and demonstrate its effectiveness and utility by incorporating it as the core component in 
a novel solver for the normalized cut graph clustering objective. %

Our extensive experiments on a variety of large graphs show that our expander-based algorithm 
outperforms state-of-the-art solvers for normalized cut
with respect to solution quality by a large margin on a variety of graph classes such as citation, e-mail, and social networks or web graphs
while remaining competitive in running time.

\end{abstract}



\received{8 February 2024}
\received[accepted]{17 May 2024}

\maketitle

\def\nvol{\operatorname{vol}}
\def\border{\operatorname{border}}
\def\ncut{\operatorname{ncut}}

\section{Introduction}
In recent years, expander decompositions and expander hierarchies
have emerged as fundamental tools within the theory community
for developing fast approximation algorithms and fast dynamic 
algorithms for such diverse problems as, e.g.,
routing~\cite{haeupler2022hop}, connectivity~\cite{expanderHierarchy}, 
(approximate) maximum flow~\cite{DBLP:conf/soda/KelnerLOS14, DBLP:conf/focs/ChenKLPGS22},
or triangle enumeration~\cite{DBLP:journals/jacm/ChangPSZ21}.
Informally, an expander is a well-connected graph. Its \emph{expansion} or
\emph{conductance} $\phi = \min_{S \subseteq V} \frac{\border(S)}{\min(\nvol(S), \nvol(\bar{S}))}$ is the minimum ratio 
between the number of edges leaving any subset of vertices and the number of
edges incident to vertices of the subset. A large value of $\phi$ (close to $1$)
indicates a graph in which no subset of vertices can be easily disconnected from the rest of the graph. %

An \emph{expander decomposition} of a graph partitions the vertices such that the  subgraph induced by each part is an expander and the number of edges between different components of the partition is low.\footnote{This can be seen as a special form of a graph clustering.} Many results over the years~\cite{spielman2004nearly,wulff2017fully,nanongkai2017dynamic,nanongkai2017dynamicB,saranurakwang19} have demonstrated that such a decomposition not only exists for every graph, but can be computed in near-linear time.
At a high level, the success of expander decomposition-based algorithms is due
to the fact that many problems are \enquote{easy} on expanders: One first
identifies regions of a graph where a problem is easy to solve (the expanders),
solves the problem (or a suitable sub-problem) within each region, and then 
combines the individual solutions to obtain a solution for the overall problem.

\emph{Expander hierarchies}~\cite{expanderHierarchy} apply expander
decompositions in a hierarchical manner. This is done by computing an expander decomposition of the graph and then contracting the individual expanders into single vertices repeatedly. This process is repeated until the entire graph has been contracted to a single vertex.
This recursive decomposition has a corresponding decomposition tree $T$ (the
so-called \emph{(flow) sparsifier}) where the root corresponds to the entire graph, the
leaves correspond to vertices of the original graph, and inner vertices
correspond to the expanders found during the decomposition procedure.
The
sparsifier approximately preserves the cut-flow structure of the original
graph in a rigorous sense. This relationship has been exploited in 
numerous applications in static and dynamic graph
problems~\cite{hua2023maintaining,li2021deterministic,haeupler2022hop} to
obtain fast algorithms with provable guarantees.

While many algorithms for expander decomposition that offer strong asymptotic bounds on the running time have been
suggested~\cite{spielman2004nearly,wulff2017fully,nanongkai2017dynamic,nanongkai2017dynamicB,saranurakwang19},
\footnote{Including the near-linear-time algorithm by Saranurak and
Wang~\cite{saranurakwang19}, which finds a decomposition where each component
has expansion at least $\phi$ in time $\bigosym(m \phi^{-1} \log^4 m)$.} practical algorithms based on expander decomposition have not seen success thus far. This is due to these
algorithms requiring one to solve many maximum flow problems, which means these algorithms are prohibitively slow in practice for graphs with many edges.
Arvestad~\cite{arvestad2022near}, e.g., reports that decomposing a graph with
approximately $10^5$ edges can take almost \SI{5}{\minute} for various values of $\phi$ in their implementation of~\cite{saranurakwang19}. %

In practice, the \emph{multilevel graph partitioning} framework has been the de-facto approach for computing high quality solutions for cut problems on graphs. This framework consists of a \emph{coarsening} step, in which a smaller representation of the graph is computed, a \emph{solving} step, where a solution is computed on the coarse graph and a \emph{refinement} step, during which the solution is improved using local methods during the graph uncoarsening process.
Multilevel graph partitioning has been successfully applied for computing balanced cuts~\cite{metis,kahip,osipov2010n,auer2012graph} and normalized cuts in graphs~\cite{graclus}. 

We note that algorithms based on the expander hierarchy approach may also be
regarded as a variant of multilevel graph partitioning. In this approach, a
graph is first coarsened by recursive contractions to obtain a smaller graph 
that reflects the same basic cluster structure as the input graph. Then an initial
partition is computed on the sparsifier and, afterwards, in a series of
refinement steps the solution is mapped to the input graph while improving it locally as we undo the coarsening.
Despite the
theoretical guarantees brought forward by expander decompositions and expander
hierarchies, the latter have not yet led to similar breakthroughs in the field
of fast, high-quality experimental algorithms. The reason is that the expander
decomposition required at each coarsening step becomes a significant
performance bottleneck in practice. In contrast, multilevel graph partitioners
like \Metis{}~\cite{metis} and \Kahip{}~\cite{kahip} use a fast matching
approach here.

In this work we mitigate the
computational bottleneck by introducing a novel, practically efficient
random-walk-based algorithm for expander decomposition. Based on this, we give
the first implementation of the expander hierarchy and, thus, an algorithm to
compute tree flow sparsifiers, allowing us to solve various cut problems on
graphs effectively. Exemplarily, we show that our approach is eminently
suitable to compute normalized $k$-cuts on graphs, where the goal is to
partition the vertex set into $k$ clusters such that the sum over the number of
edges leaving each cluster, normalized by the cluster's volume, is
minimized. %
Normalized cut is a popular graph clustering objective and particularly able to
capture imbalanced clusterings.
Its manifold applications include, e.g., community
detection~\cite{DBLP:conf/www/LiZHRCH20} and
mining~\cite{DBLP:conf/kdd/CaiSHYH05}, topic
reconstruction~\cite{DBLP:journals/jmma/Bichot10}, story
segmentation~\cite{DBLP:conf/airs/ZhangXFZ09},
bioinformatics~\cite{graclus,DBLP:conf/ismb/XingK01}, tumor
localization~\cite{sahoo2023brain}, and image
segmentation~\cite{DBLP:journals/pami/ShiM00,tokencut}. %
It is closely related to spectral clustering, which uses spectral properties of
eigenvalues and eigenvectors of the graph's
Laplacian. %
However, the spectral approach suffers both from very large running times and
memory requirements to compute and store the eigenvectors, and thus does not
scale well~\cite{graclus}. This problem was addressed by Dhillon
\ea{}~\cite{graclus} as well as Zhao \ea{}~\cite{zhao2018nearly}, who presented
algorithms for normalized cut that either use spectral methods only after
coarsening~\cite{graclus} or apply them on a spectrally sparsified
graph~\cite{zhao2018nearly}.

\paragraph{Contributions.}
We present the first practically efficient algorithm for expander decomposition and the first implementation to compute an expander hierarchy.
Our approach is based on random walks and is justified by rigorous theoretic and empirical analysis.
We report on a comprehensive experimental study on normalized cut solvers, comprising \num{50}~medium-sized to very large graphs of various types, where we compare our  expander-based algorithm, \Algo{}, 
to \Graclus{} by Dhillon \ea{}~\cite{graclus}, the approach by Zhao \ea{}~\cite{zhao2018nearly}, as well as the state-of-the-art graph partitioners \Metis{} and \Kahip{}, which do not specifically optimize towards the normalized cut objective, but are fast and in practice often used for this task.
The experiments show that our algorithm produces superior normalized cuts on graph classes such as citation, e-mail, and social networks, web graphs, and generally scale-free graphs, and is only slightly worse on others.
On average, it is still distinctly the best across all graphs and values of $k$.

If only a single value of $k$ is desired, its running time is on average only \num{3}~times slower than the runner-up, \Graclus{}, and never exceeded \SI{18}{\minute}.
A notable advantage of \Algo{} is that it can quickly compute solutions for multiple values of $k$ once a sparsifier is computed, which can be faster than running \Graclus{} multiple times.

\section{Related Work}\label{sec:related}
Our work is motivated by recent theoretical results~\cite{haeupler2022hop,expanderHierarchy,DBLP:conf/soda/KelnerLOS14,DBLP:conf/focs/ChenKLPGS22,DBLP:journals/jacm/ChangPSZ21}
building on expander
decompositions~\cite{spielman2004nearly,wulff2017fully,nanongkai2017dynamic,nanongkai2017dynamicB,saranurakwang19} and the expander hierarchy~\cite{expanderHierarchy} and inspired by non-spectral
approaches~\cite{graclus} to tackle the normalized cut problem.

Mohar~\cite{mohar1989isoperimetric} showed that the computation of the 
so-called isoperimetric number or conductance of a graph
(see \autoref{sec:pre})
is NP-hard, which implies the hardness of the normalized $k$-cut problem for $k \geq 2$.
Normalized cut remains NP-hard on weighted trees~\cite{daneshgar2012complexity}.

A number of tools that have been used to solve normalized cut use spectral methods~\cite{DBLP:conf/nips/NgJW01,DBLP:conf/iccv/YuS03,nvidia,DBLP:journals/siamcomp/Peng0Z17,DBLP:conf/ijcai/0004NHY17,DBLP:journals/pami/NieLWWL24}.
This usually requires to compute $k$ eigenvectors of the Laplacian matrix, which was shown to scale badly in practice~\cite{graclus,zhao2018nearly}.
Afterwards, an additional discretization step is necessary to obtain the clustering.

Dhillon~\ea{}~\cite{graclus} therefore suggest an algorithm called \Graclus{},
which is based on the multilevel graph partitioning framework. It applies the
same coarsening steps as \Metis{}, but with a modified matching procedure. The
coarsened graph can then be partitioned using different approaches, including a
spectral one.
In the refinement step, \Graclus{} uses a kernel $k$-means-based local search algorithm for improving the normalized cut objective value.
The authors evaluate their algorithm experimentally against \Metis{} as well as a spectral
clustering algorithm~\cite{DBLP:conf/iccv/YuS03}.
They show that it outperforms the spectral method w.r.t.\ normalized cut value, running time, and memory usage. 
It also produces better results than \Metis{} and is comparable w.r.t.\ running time.

Zhao~\ea{}~\cite{zhao2018nearly} employ a joint spectral sparsification and coarsening scheme to produce a smaller representation of the graph that preserves the eigenvectors of the Laplacian in near-linear time $\tilde{O}(m)$.
Afterwards, a normalized cut is computed on the sparsifier using spectral clustering.
Their sparsification scheme obtains a significant reduction in the number of edges and nodes, which makes it feasible to apply the spectral method on the reduced graphs, without the quadratic term in the running time growing too large.
The authors again perform an experimental comparison with \Metis{} and observe that their algorithm overall outperforms \Metis{} w.r.t.\ the normalized cut value, while being slightly slower.

\Metis{}~\cite{metis}, \Kahip{}~\cite{kahip}, and many other graph partitioning tools~\cite{jostle,chaco} do not solve the normalized cut problem and instead aim to find good solutions to a balanced graph partitioning problem, where the vertex set is to be partitioned into $k$ sets of (roughly) equal size, while minimizing the number of edges cut. 
\textsf{CHACO}~\cite{chaco} implements a spectral graph partitioning approach, but the number of clusters is limited to at most $k = \num{8}$.

Nie~\ea{}~\cite{DBLP:journals/pami/NieLWWL24} recently presented a spectral normalized cut solver based on the coordinate descent method along with several speedup strategies.
They evaluate their algorithm against two other spectral methods~\cite{DBLP:conf/nips/NgJW01,DBLP:conf/ijcai/0004NHY17} on a number of medium-sized data sets and show that it consistently computes the best solution and is the fastest.
A notable difference is that $k$ is not an input parameter.

\section{Preliminaries}\label{sec:pre}

For an undirected graph $G=(V,E)$ we use $d_v$ to denote
the degree of vertex $v\in V$, $\smash{\vec{d}}$ to denote the \emph{degree
  vector}, i.e., the vector of vertex degrees, and $D=I\smash{\vec{d}}$ to denote the
corresponding \emph{degree} matrix. $\Delta$ is used to denote the maximum
degree of a vertex in $G$. For a subset $S$ of vertices we define the
\emph{volume} $\nvol(S)$ as the sum of vertex degrees of vertices in $S$.
Its \emph{border} $\border(S)$ (or \emph{capacity}) is defined as $|E(S,\bar{S})|$, where
$E(X,Y)=\{\{x,y\}\in E\mid x\in X,y\in Y\}$ is the set of edges between subsets
$X,Y\subseteq V$, and $\bar{S}=V\setminus S$.

A $\emph{$k$-cut}$ is a partition $\mathcal{P}$ of the vertex set into $k$
non-empty parts. Given such a partition $\mathcal{P}=(S_1,\dots,S_k)$, its
\emph{normalized cut value} $\nc$ is defined as
\begin{equation*}
\nc(\mathcal{P}) = \sum_{i=1}^{\smash{k}} \frac{\border(S_i)}{\nvol(S_i)}\enspace.
\end{equation*}

\noindent
For $k=2$ we will usually specify a $2$-cut (or just cut) by its
smaller side, i.e., we will refer to the cut $(S,\bar{S})$ by just $S$, where
$\nvol(S)\le\nvol(\bar{S})$. The \emph{conductance} of a ($2$-)cut is defined
as $\Phi(S)=\border(S)/\nvol(S)$ and the \emph{conductance of a graph}
$G=(V,E)$ is $\Phi(G)=\min_{\text{cuts $S$}}\Phi(S)$.\footnote{If there are no
  cuts in $G$ (i.e., $|V|=1$) we define its conductance to be 1.} The problem
of finding a $2$-cut with minimum conductance is closely related to the problem
of finding a $2$-cut with minimum normalized cut value, because
for any $2$-cut, $\Phi(S) \le \nc(S) \le 2\Phi(S)$. The normalized $k$-cut
objective can be seen as 
a generalization of conductance to $k$-cuts.

To properly describe our random walks we need some further notation. For a cut
$S$ we call $\nvol(S)/\nvol(V)$ the \emph{balance} of the cut $S$ and denote it
with $b(S)$. For a subset $A\subseteq V$ we use $G\{A\}$ to denote the subgraph
induced by $A$ with self-loops added so that the vertex degrees do not change
(a self-loop counts 1 to the degree of a vertex). This means the degree of a
vertex $a\in A$ is the same in $G$ as in $G\{A\}$. Whenever the graph in
question is not clear from the context we use subscripts to indicate the graph,
i.e., we write $\nvol_G(S)$, $\Phi_G(S)$, $E_G(X,Y)$, etc. To avoid
  notational clutter we write $\nvol_A(S)$, $\Phi_A(S)$, $E_A(X,Y)$ when we are
  referring to the graph $G\{A\}$.

We say a graph $G=(V,E)$ is a \emph{$\phi$-expander} if $\Phi(G) \ge \phi$, and we call a
vertex partition $(V_1,\dots,V_\ell)$ a \emph{$\phi$-expander decomposition} if $\Phi(G\{V_i\}) \ge
\phi$ for all $i$.

\section{Expander Decomposition using Random Walks}\label{sec:expdec}
\def\nvol{\operatorname{vol}}

\def\tildeO{\tilde{\mathcal{O}}}

In this section we describe our random-walk-based approach for obtaining
expander decompositions and present its theoretical guarantees. 
The complete proofs
 The complete proofs for these guarantees can be found in the appendix. %

A natural approach for computing expander decompositions is to find a low
conductance cut to split the graph into two parts and then to recurse on both
sides. If no such cut exists, we have certified the (sub-)graph to be an
expander.
In the end, the whole procedure terminates if each subgraph is an expander,
i.e., it terminates with an expander decomposition. Saranurak and
Wang~\cite{saranurakwang19} used this general approach to obtain an expander
decomposition that runs in time $\mathcal{O}(m\log^4{m}/\phi)$ and only cuts $\mathcal{O}(\phi m\log^3{m})$ edges. While their flow-based techniques
give very good theoretical guarantees, the hidden constants do not seem to allow
for good practical performance (see e.g.~\cite{arvestad2022near}).

In this work we base the cut procedure of the expander decomposition on random
walks. As a consequence, we can only guarantee that our decomposition cuts at
most $\tildeO( \sqrt{\phi}\,m )$ edges\footnote{$\tildeO$ hides polylogarithmic factors.}%
, since we are limited by the intrinsic
Cheeger barrier of spectral methods. However, random walks have a very simple
structure, which leads to a simple algorithm with good practical performance.
The weaker dependency on $\phi$ ($\sqrt{\phi}$ instead of $\phi$) is not crucial
for our graph partitioning application because when using the expander
decomposition to build an expander hierarchy one chooses $\phi$ as large as
possible anyway.

\begin{theorem}[Expander Decomposition]\label{theorem:expander-decomposition}

Given a graph $G$ with $m$ edges and a parameter $\phi$, there is a random-walk-based algorithm that with high probability finds a $\phi$-expander
decomposition of $G$ and cuts at most
$\bigosym (\sqrt{\phi}\, m \log^{5/2}{m} )$ edges. The running time is
$\bigosym \big( (m\log^{7/2}{n})/\phi^{5/2} \big)$.
\end{theorem}

The main part of our algorithm is the cut procedure described in
Section~\ref{sec:goodcuts}. This procedure is then plugged into the framework of
Saranurak and Wang to find an expander decomposition. Our cut procedure gives
the following guarantees.

\begin{theorem}[Cut Procedure]\label{theorem:cut-step}
Given a graph $G = (V,E)$ with $m$ edges and a parameter $\phi$, the cut
procedure takes $\bigosym \big((m \log{n})/\phi^2 \big)$ steps and terminates with
one of these three cases:
\begin{enumerate}
    \item We certify that $G$ has conductance $\Phi(G) \geq \phi$.
    \item We find a cut $(A, \bar{A})$ in $G$ that has conductance at most $\Phi_G(A) = \bigosym (\sqrt{\phi}\log^{3/2}{m} )$. Then one of the following holds:
    \begin{enumerate}
    \item either $\nvol(A), \vol{\bar{A}}$ are both $\Omega(m \sqrt\phi / \log^{3/2}{m})$,
    i.e., $(A, \bar{A})$ is a relatively balanced low conductance cut;
        \item or $\vol{\bar{A}} = \bigosym (m\sqrt\phi / \log^{3/2}{m})$ and A is a near $6\phi$-expander.
    \end{enumerate}
\end{enumerate}
\end{theorem}

Given the cut procedure, an expander decomposition is computed as follows.
On the current subgraph $G$, execute the cut procedure, to either find a low
conductance cut $S$ or certify that none exists. If no such cut exists, then
$G$ is a certified $\phi$-expander and we terminate. Otherwise we check whether
$S$ is sufficiently balanced, i.e., the volume of the smaller side is at least
$\Omega(m \sqrt\phi/ \log^{3/2}{m})$. In that case we cut the edges across
$(S, \bar{S})$ and recurse on both parts. As both
parts are substantially smaller than $G$ we obtain a low recursion depth.
Otherwise, $S$ is very unbalanced but the larger side $\bar{S}$ is a so-called \emph{near
  expander} -- a concept introduced by Saranurak and Wang~\cite{saranurakwang19}:
\begin{definition}[Near $\phi$-expander]
Given a graph $G=(V,E)$. A subset $A\subseteq V$ is a near
$\phi$-expander in $G$ if for all sets $ X\subseteq A$ with $\nvol(X)\le \nvol(A)/2:\,|E(X,V\setminus X)|\ge\phi\nvol(X)$.
\end{definition}
Note that if the LHS in the above equation were $|E(X,A\setminus X)|$ then $G\{A\}$
would be a $\phi$-expander.

If $G\{\bar{S}\}$ (for $\bar{S}$ returned by the
cut-procedure) is indeed a $\phi$-expander we can
just recurse on the smaller side $S$ and would obtain a low recursion depth.
Saranurak and Wang introduced a trimming procedure that, given a subset $B$ that
is a near $6\phi$-expander with volume $\nvol(B)\ge 9\nvol(V)/10$ and
$|E(B,\bar{B})|\le\phi\nvol(B)/10$, computes a subset $B'\subseteq B$
that is a proper $\phi$-expander and has volume $\nvol(B')\ge\frac{1}{2}\nvol(B)$.
By applying this trimming step to $\bar{S}$ we can return $G\{\bar{S}'\}$ as a
proper $\phi$-expander and recurse on the remaining graph -- still with a small
recursion depth.
Due to the slightly weaker guarantee of the cut procedure, we need a smaller balance factor of $\beta = \Theta(\sqrt\phi/\log^{3/2}{m})$ to ensure that $|E(B,\bar{B})|\le\phi\nvol(B)/10$ and apply the Trimming procedure.
Overall, using our cut procedure within this framework gives a recursion depth of $O(\log^{5/2}{m}/\sqrt\phi)$ and yields 
Theorem~\ref{theorem:expander-decomposition}.

\subsection{Finding Low Conductance Cuts}\label{sec:goodcuts}
\def\bal{{\color{red}\operatorname{bal}}}

We now give a detailed description of the cut procedure that forms the basis of
Theorem~\ref{theorem:cut-step}. The goal is to either certify that $G$ is a
$\phi$-expander or to find a low conductance cut that is as balanced as possible.
The idea is to exploit that random walks converge quickly on expanders and
hence when they don't, we know there must be a low conductance cut. See Algorithm~\ref{alg:cut-step} for an outline.

We employ a  concurrent random walk, where each node distributes its unique commodity in the graph.
We are interested in the probability that after t steps a particle that started say at node i is at some other node j. The walk has converged if this distribution is essentially identical for \emph{every} starting vertex.
If we quickly reach this stationary distribution, there cannot be a low conductance cut and hence the graph must be an expander. Otherwise we can use information gathered from the walk to find a low conductance cut.

{%
\def\uni#1{\mathrlap{#1}\hspace*{0.83em}}
\begin{algorithm}
\caption{Cut Procedure}\label{alg:cut-step}
\begin{algorithmic}[1]
\linespread{1.08}\selectfont
\Require Graph $G=(V,E)$, Target expansion $\phi$
\Ensure \textsc{Expander} or \textsc{Balanced}$(S,\bar{S})$ or \textsc{Unbalanced}$(S,\bar{S})$

\State $\uni{T}      \gets1/(12\phi)$
\State $\uni{\gamma} \gets 343\sqrt{\rule{0pt}{1.8ex}\smash{\phi \log(32 m^3)\log^2(n) }}$\Comment{{\color{Gray}cut threshold}}
\State $\uni{\beta}  \gets 2\sqrt\phi/\log^{3/2}m, $\Comment{{\color{Gray}balance}}
\State $W  \gets I$\Comment{{\color{Gray} random walk matrix}}
\State $A \gets V$, $L \gets \emptyset$

\For{$t = 1, \, \ldots, \, T$}

    \State $\uni{\vec{r}} \gets \textit{random unit vector in } \mathbb{R}^{n}$ %

    \State $\uni{\vec{u}} \gets W D^{-1} r$
    \Comment{{\color{Gray} apply random walk}}
    
    \State $\uni{\vec{u}} \gets \vec{u} - (\vec{u}^\top \vec{d} \,)/\vol{V} \cdot \bbone $ \Comment{{\color{Gray} ensure $\vec{u}\perp\vec{d}$}}

    \State $\mathcal{A}$, $D_{\!A} \gets$ \textit{adjacency and degree matrix of} $G\{A\}$

    \State $W \gets \big(\frac{1}{2}I+\frac{1}{2}\mathcal{A}D_{\!\!A}^{-1} \big) W$
    \Comment{{\color{Gray} extend walk matrix}}

    \State $\uni{S} \gets \text{SweepCut}(\vec{u}, \gamma)$
    
    \If{$\nvol(S) \geq \beta m$}
            \Return \textsc{Balanced}$(S, \bar{S})$ 
    \ElsIf{$S \mathbin{\smash{\neq}} \emptyset$} %
        \State $L \gets L \cup S$, $A \gets A \setminus S$
        \If{$\vol{L} \geq \beta m$}
            \Return \textsc{Balanced}$(A, L)$ 
        \EndIf
    \EndIf
\EndFor

\If{$L = \emptyset$}
    \Return \textsc{Expander}
\Else{}
    \Return \textsc{Unbalanced}$(A, L)$
\EndIf
\end{algorithmic}
\end{algorithm}}

Such a cut may however be very unbalanced.
The procedure therefore accumulates low conductance cuts until the combined cut is a
balanced low conductance cut or the graph that remains does not have a low
conductance cut anymore. Here we call a cut $S$ \emph{balanced} if its balance
$b(S) \ge \beta \defeq 2\sqrt\phi/\log^{3/2}m$.

The algorithm maintains a partition of $V$ into two sets $A, L$ for each iteration $t$ with initial values $A := V, L := \emptyset$.
We repeat the following for $T=1/(12\phi)$ steps:
In iteration $t$, we generate a new random unit vector $r$
and execute $t-1$ steps of the random walk, initialized according to $r$.
This walk yields a vector $u$, on which we analyze the conductance of all sweep cuts, i.e., cuts of the form
$S_c \defeq \{ v \in A : u_v \leq c \}$ for some value $c$.
Note that these conductance values can be calculated in linear
time after sorting the entries of $u$.

We consider a cut to have low conductance if the value is below the threshold $\gamma = \mathcal{O}(\sqrt{\phi} \log^{3/2}{m})$. A lower threshold value $\gamma$ would give better guarantees in case (2) of Theorem~\ref{theorem:cut-step} but at the same time it would increase the number of rounds required to converge and hence worsen the guarantee of case (1) in Theorem~\ref{theorem:cut-step}. The value is chosen to ensure we can guarantee $\Phi(G) \geq \phi$ in case (1) of the Theorem.

The analysis of the sweep cuts yields one of these three cases:
\begin{enumerate}
    \item If all sweep cuts have conductance at least $\gamma$, we continue
    with the next iteration.
    \item If there exists a sweep cut with conductance $<\gamma$ and balance $\ge\beta$ we return this low conductance cut. 
    \item Otherwise we consider the two-ended sweep cuts of the form $S_{a,b} \defeq \{ v \in A : u_v \notin [a, b] \}$ for values $a\leq b$ and find the one with largest volume among those with $\Phi_{A}(S_{a,b})<\gamma$.
    If this cut has balance $\ge\beta$ we return it, otherwise we move it from $A$ to $L$ and check whether $L$ has become balanced.

\end{enumerate}

The random walk can be interpreted as the projection of a much higher dimensional random walk onto the randomly chosen direction $r$. This projection step is crucial for the implementation to become computationally feasible and we show that the projections approximate the original structure sufficiently well.

To argue the correctness of the cut procedure we have to show that it
is highly unlikely that the procedure does not find a cut on a graph that has
expansion less than $\phi$. For this we argue that after $T$ random walk
steps (without finding a cut) the walk will have \enquote{converged} to its
stationary distribution w.h.p. Because of the choice of parameters such a quick
convergence is only possible if $G$ is a near $6\phi$-expander. Whenever the
cut procedure returns a cut, it is guaranteed to have conductance at most $\gamma$. From this it
follows that the expander decomposition cuts at most $\mathcal{O}(\gamma m\log m)$ edges.
 The entire argument can be found in Section~\ref{sec:proof} in the appendix. %

\section{\texorpdfstring{\Algo{} -- A New Normalized Cut Algorithm}{XCut - A New Normalized Cut Algorithm}}\label{sec:algorithm}
In this section, we introduce the algorithm \Algo{}, which is based on the previous section's novel random walk-based expander decomposition.
We note the apparent similarity between multilevel graph partitioning and the expander hierarchy and use this as the basis of \Algo{}.
As an outline, we use the novel random walks to construct the expander hierarchy to obtain a coarse representation of the graph, the tree flow sparsifier.
We then compute an initial solution on the tree.
Finally, we use an iterative refinement step while descending the hierarchy to improve the solution we found.
Compared to other contraction schemes, which lead to each vertex in the coarsest graph representing roughly the same number of nodes in the base graph, the subtrees on each level in the tree flow sparsifier can represent a vastly different number of vertices.

\emph{Expander Decomposition.}
While the expander decomposition outlined in \autoref{sec:expdec} is much
simpler to implement than that of~\cite{saranurakwang19}, as we do not rely on
maximum flow computations at all, we made several choices in the implementation to
speed up computation. We iterate a single random walk, and after each iteration, we
check whether we can find a sparse cut. If we find a suitable cut, we
disconnect the edges going across it, but contrary to the algorithm in
\autoref{sec:expdec}, we do not restart the random walk, as we find we can extract
further information about the cut structure of the graph from the state of the
random walk. For example, if the random walk has mixed very well on one of the
new components, it is likely to be an expander, while if there is another
sparse cut in the component, then the random walk will likely not have mixed
well on the component. One may think that reducing the number of random walks might lead to a loss
of guarantees and higher variance of the algorithm, but in experiments conducted while designing the algorithm, we found that on real graph instances running multiple concurrent random walks is not necessary. In fact, only a single graph in our 50 graph benchmark had noticeable variance. See also the discussion in
\autoref{sec:graclus}.

The main parameter of the expander decomposition is the cut value $\gamma$, which is the minimal sparsity of the cuts our algorithm makes.
Additionally, we introduce a parameter $\rho$, to be used as a threshold for ``certifying" that a component is an expander, as we found that choosing the threshold to be $1/(4\nvol(V)^2)$ does not offer any benefits over a much larger value.
See also \autoref{sec:configuring} for details on choosing the value for this parameter. We make a final modification to the theoretical algorithm in that we omit the trimming step on unbalanced cuts, since it does not provide any further speedup of the expander decomposition routine in practice.

\emph{Automatically Choosing $\gamma$.}
The theoretical analysis in~\cite{expanderHierarchy} suggests a choice %
for $\gamma$ and $\phi$ that is sufficient to prove the theoretical results.
In preliminary experiments we found this choice to be too pessimistic and, in fact,
by adapting $\phi$ and $\gamma$ to the graph we can obtain better results.
However, there is a trade-off to be made:
If $\gamma$ is too small, the expander decomposition will not find
many sparse cuts, and we obtain sparsifiers of low quality. %
On the other hand, if $\gamma$ is too large, most cuts will be sparser than $\gamma$, which leads to many cuts
being made and increases the running time. In the worst case, this can even
prevent the algorithm from terminating. %

Thus, our goal is to choose $\gamma$ such that it offers a good quality vs.\ time trade-off. %
When choosing $\gamma$, we can only observe whether this was a good choice in a post hoc fashion.
A naive strategy would be to start with a large $\gamma$ and decrease it until the expander hierarchy terminates within a reasonable amount of time.
However, this approach is wasteful, as we discard previously computed decompositions, even if they were good.
Instead, we decrease $\gamma$ by multiplying it with constant factor $\epsilon < 1$ whenever the expander decomposition on a specific level cuts too many edges in $G_i$.
We then use the new $\gamma' = \epsilon \gamma$ for the remaining expander decompositions, decreasing it further as required.

\emph{Solving on the Sparsifier.}\label{sec:heuristic} To solve normalized $k$-cut on the tree sparsifier obtained from the hierarchy, we want to remove $k-1$ edges to decompose it into a forest of $k$ trees.
Given a solution, i.e., a tuple of edges $(e_1, \dots, e_{k-1})$, we assign vertex $v \in V$ to the cluster $C_j$ associated with $e_j$ if $e_j$ is the first edge we encounter on the path from $v$ to the root.
If no edge in the solution lies on the path to the root, we assign $v$ to cluster $C_k$.
As normalized cut is NP-hard also on trees (see~\autoref{sec:related}), %
we introduce two heuristic approaches that take time $\bigosym(nk)$ each: 

\emph{Greedy:} 
The simple greedy heuristic picks the edge in the sparsifier which minimizes the increase in the normalized cut objective in every step.
By simply computing the cost of cutting each remaining edge, it takes $O(n)$ time to find this edge, assuming the number of vertices in the sparsifier is $\bigosym(n)$.
To partition a graph into $k$ clusters, we repeat this process $k-1$ times. 
For $k=2$, this algorithm produces the optimal solution on the tree as we pick the edge that minimizes the cut objective.

\emph{Dynamic Programming:} Each row of the dynamic program corresponds to a level, and we make one cell for each pair $(v, i)$ of the row, where $v$ is a vertex in the sparsifier and $i \in [0, k]$. 
The value of each cell is the normalized cut value $\nc$ of decomposing the subtree rooted at $v$ into $i$ parts.
We write $DP(v, i)$ for this value. 
Additionally, we write $cut(v, i)$ for the weight of cut edges in the solution incident to the subtree rooted at $v$, as well as $vol(v, i)$ for the volume remaining in the subtree.
    
Without loss of generality, assume the tree is binary, as otherwise we can binarize the tree by inserting edges of infinite cost (see Henzinger~\ea{}~\cite{dpSTACS} for details). 
The value of a cell is computed according to the following rule:
\def\DP{\operatorname{DP}}
\def\cut{\operatorname{cut}}
\def\cutParent{\operatorname{cutParent}}
\begin{align*}
    \DP(v, j) =\min(&\cutParent(\DP(v, j-1)), \\
                   &\min_{0 \le i \le j} (\DP(v_l, i) + \DP(v_r, j - i))),
\end{align*}
where $\cutParent(\DP(v, j-1)) = \DP(v, j-1) + \frac{w(v, v') + \cut(v, j-1)}{\nvol(v, j-1)}$ is the best solution where the edge going to the parent is cut.
For each vertex $v$ of $G_0$, we initialize the bottom row of the program with $\DP(v, 0) = 0$, $\DP(v, 1) = 1$ and $\DP(v, j) = \infty$ for $j \in [2, k]$.
In the root vertex we use the special rule $\DP(v, j) = \min_i \DP(v_l, i) + \DP(v_r, j - i) + \frac{\cut(v_l, i) + \cut(v_r, j-i)}{\nvol(v_l, i) + \nvol(v_r, j-i)}$, where the last term ensures we do not produce a solution with $\nvol(v, i) = 0$, leading to cluster $j$ being empty.

\begin{figure}\Description{A scatterplot depicting the ratio between maximum and median degree of real world graphs on the x axis and the relative value of XCut's normalized cut versus Graclus is shown on the y axis. The points show a negative trend, decreasing in value as the ratio increases.}
    \centering
    \includegraphics[width=.9\linewidth]{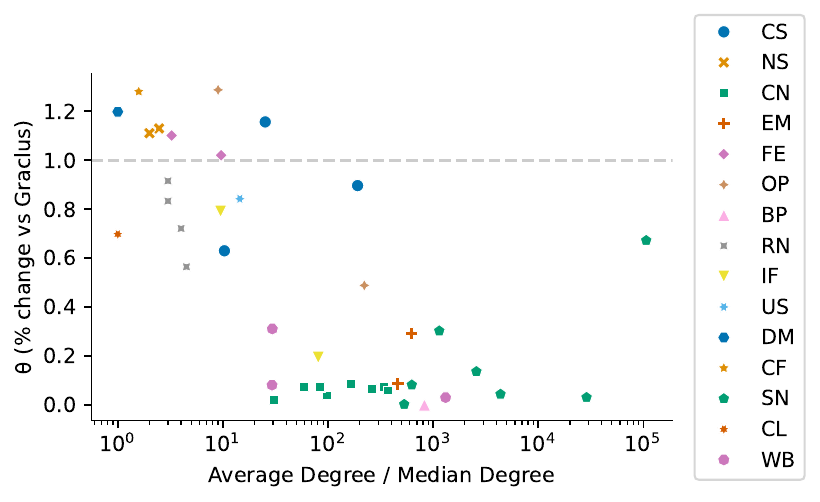}
    \caption{\mathversion{bold}
    Relative improvement over \Graclus{} (y-axis) vs.\ the ratio of the maximum degree and the median degree (x-axis). The value on the x-axis is larger if the graphs exhibit a distribution with large outliers. Note the negative trend, except for the outlier towards the right corresponding to instance SN7.
    }\label{fig:max-median}
\end{figure}

\emph{The Refinement Step.}
Finally, we perform an iterative refinement as we descend the hierarchy. While descending, we introduce new clusters according to the edges in the solution. On each level we then perform vertex swaps that improve the normalized cut objective, by either reducing the number of cut edges or making the partition more balanced. For details see the full version of the paper. %

\emph{Support for Variable Number of Partitions~$k$.}
A notable strength of our algorithm is that with both heuristics, it solves the problem for any value $k' \le k$ during their execution. Suppose we are still determining the number of clusters needed.
In that case, we can compute the solution for the maximum $k$ we are interested in and obtain solutions for all smaller numbers of partitions during exploratory data analysis.
The only step that needs to be rerun is the refinement step.

\section{Experimental Evaluation}\label{sec:experiments}
In the previous sections, we have shown that our approach provides provable guarantees on the approximation ratio for the value $\nc$ of the normalized cut\footnote{See Section~\ref{sec:pre} for the precise definition.} (and its relatives, sparse cut, and low-conductance cut) if $k = 2$.
We now turn to evaluate \Algo{} experimentally in different configurations.
We compare the objective value of normalized cuts produced by \Algo{} and its running time against the normalized cut solver \Graclus{} by Dhillon \ea{}~\cite{graclus} and the state-of-the-art graph partitioning packages \Metis{}~\cite{metis} and \Kahip{} (kaffpa)~\cite{kahip}, all of which are available publicly. 
These algorithms are based on the multilevel graph partitioning framework %
and produce disjoint partitions of the vertices into $k$ clusters, where $k$ is a freely choosable parameter.
We note that \Metis{} and \Kahip{} solve the balanced $k$-partitioning problem rather than normalized cut. Nevertheless, we include these solvers in our comparison, as they are used for this
task in practice and we found that they can outperform \Graclus{} on some of the graphs in our benchmark dataset.
By this, we follow the methodology of \cite{graclus} and \cite{zhao2018nearly}.

In addition, we compare our results to the values reported for the normalized cut algorithm by Zhao et al.~\cite{zhao2018nearly}\footnote{Unfortunately, the code is not available publicly and also could not be provided by the authors upon request before the submission deadline.}.
We omit comparisons to solvers employing spectral methods such as the recent works by Chen~\ea{}~\cite{DBLP:conf/ijcai/0004NHY17} and Nie~\ea{}~\cite{DBLP:journals/pami/NieLWWL24}, as this approach does not scale well to large datasets of millions of nodes~\cite{graclus}.
In preliminary experiments we found that the solver of Nie et al.~\cite{DBLP:journals/pami/NieLWWL24} uses over 330GB of memory on instance CN3, whereas our solver used less than 400MB of memory.
Furthermore, the algorithm presented in \cite{DBLP:journals/pami/NieLWWL24} does not necessarily produce $k$ clusters, thus making direct comparisons difficult.
Lastly, the space complexity of spectral methods becomes prohibitive for larger values of $k$, as noted by~\cite{graclus}.

\subsection{Experimental Setup}
\paragraph{Instances.}
Our setup includes \num{50} graphs from various applications. 
See \autoref{tab:graph-types} for an overview and the full version of the paper for details.
To facilitate comparability, our collection contains the eight instances used by Dhillon \ea{}~\cite{graclus} (BP1, CF1, CS1--3, DM1, OP1, and OP2)
as well as the \num{21}~instances used by Zhao \ea~\cite{zhao2018nearly}.
In addition, we selected \num{21} real-world networks that cover various application areas, including some of larger sizes.
All instances are available publicly in the Network Repository~\cite{nr} or the SuiteSparse Matrix collection~\cite{suitesparse}.

Note that a graph with $k$ or more connected components always has a (normalized) $k$-cut of size \num{0}.
Surprisingly, we found that \Algo{} is the only solver tested here that finds the trivial optimal solution if $k$ is less than the number of connected components.
However, testing connectivity before starting a solver remedies this problem, which is why we decided to exclude graphs with more than \num{128} connected components except for one (graph ID~\num{0}), which we keep for consistency reasons as it was used in previous comparisons~\cite{graclus}.

\paragraph{Methodology.}
As \Algo{}, \Kahip{}, and \Metis{} are randomized, we ran each of them $\ell = \num{10}$ times per instance with different seeds. 
\Graclus{} is deterministic, and we ran it three times to obtain a stable value for the running time.
We use the arithmetic mean over the $\ell$ runs for each instance to approximate the expected value of $\nc$ and the running time.
When reporting values, we write \Algo{mean} for the mean value across these \num{10}~runs, 
and \Algo{min} for the minimum.
As \Kahip{} and \Metis{} behaved almost identically over all $\ell$ runs, we only report mean values for them.
For each algorithm and each graph, we compute a partitioning consisting of $k \in \{2, 4, 8, 16, 32, 64, 128\}$ clusters as well as $\nc$ (smaller is better).
As a second criterion, we compare the algorithms' running times.

All experiments were conducted on a server with an Intel Xeon 16 Core Processor and \SI{1.5}{\tera\byte} of RAM running Ubuntu 22.04 with Linux kernel 5.15.
\Algo{} is implemented in \texttt{C++} and compiled using \texttt{gcc} 11.4 with full optimization\footnote{\texttt{-O3 -march=native -mtune=native}}.
For all other solvers, we followed the build instructions shipped with their code.
As \Graclus{} is single-threaded, we ran the single-threaded version of every algorithm.
\Metis{} was run in its default configuration.
We used \Kahip{} with the \texttt{fsocial} flag, which is tailored to quickly partitioning social network-like graphs\footnote{We chose this setting as we are especially interested in social network-like graphs.} and \Graclus{} with options \texttt{-l 20 -b} to enable the local search step and only consider boundary vertices during local search, %
as suggested by the authors for larger graphs~\cite{graclus}.

\begin{table}[tb]
    \caption{\mathversion{bold}
    Types and number of graphs in our benchmark dataset. 
    $\MaxDegree$ is the maximum degree, k and M are shorthand notations for \num{e3} and \num{e6},
    respectively. 
    See the full version of the paper for a detailed list.}\label{tab:graph-types}
    \setlength{\tabcolsep}{1.5pt}
    \centering
    \scalebox{0.9}{
    \begin{tabular}{lcccc}
    \toprule
        Type (Abbreviation) & \# & $|V|$ & $|E|$ & $\MaxDegree$\\ 
     \midrule
        Bipartite (\Bip) & 1 & 1.4M & 4.3M & 1.7k\\
        Computational Fluids (\ComFlu) & 1 & 17k & 1.4M & 269 \\
        Clustering (\Clus) & 2 & 4.8k-100k & 6.8k-500k & 3-17 \\
        Citation Network (\CitNet) & 9 & 226k-1.1M & 814k-56M & 238-1.1k \\
        Circuit Simulation (\CircSim) & 3 & 5k-30k & 9.4k-54k & 31-573 \\
        Duplicate Materials (\DupMat) & 1 & 14k & 477k & 80 \\
        Email Network (\Email) & 2 & 33k-34k & 54k-181k & 623-1383 \\
        Finite Elements (\FinEl) & 2 & 78k-100k & 453k-662k & 39-125 \\
        Infrastructure Network (\Infra) & 2 & 2.9k-49k & 6.5k-16k & 19-242\\
        Numerical Simulation (\NumSim) & 2 & 11k-449k & 75k-3.3M & 28-37 \\
        Optimization (\Opti) & 2 & 37k-62k & 131k-2.1M & 54-8.4k \\
        Random Graph (\RandG) & 1 & 14k & 919k & 293 \\
        Road Network (\Road) & 4 & 114k-6.7M & 120k-7M & 6-12 \\
        Social Network (\SocNet) & 7 & 404k-4M & 713k-28M & 626-106k \\
        Triangle Mixture (\TriMix) & 7 & 10k-77k & 54k-2M & 22-18k \\
        US Census Redistricting (\USCens) & 1 & 330k & 789k & 58\\
        Web Graph (\Web) & 3 & 1.3k-1.9M & 2.8k-4.5M & 59-2.6k \\
    \bottomrule
    \end{tabular}
    }
\end{table}

\subsection{\texorpdfstring{Configuring \Algo{}}{Configuring XCut}}\label{sec:configuring}
\paragraph{Greedy vs.\ Dynamic Programming (DP)}
Section~\ref{sec:heuristic} describes two heuristics for computing normalized cuts on the sparsifier.
We compare for both heuristics the running time and $\nc$ across ten precomputed sparsifiers each for 19 representative graphs. The quality returned by DP was never better than that of Greedy. While the running times for both heuristics are linear in the size of the sparsifier, the DP approach scaled worse in $k$. For example, for $k = 32$, DP was three times slower than Greedy, and seven times slower for $k = 128$. See the full version of the paper %
for a plot with the results for $\rho = 10^{-4}$. Greedy was faster and produced no worse quality than DP for all values of $\rho$. Thus, we only report results for Greedy for all further experiments.

\begin{figure}\Description{A scatterplot plotting the normalized cut value on the x axis and the running time on the y axis. There are two arrangement of dots representing two different graph instances. Different colors of points represent different choices for the threshold parameter $\rho$.
The points are arranged in L shaped patterns, with small threshold values on the top, and large threshold values on the right.}
    \centering
    \includegraphics[width=0.85\linewidth]{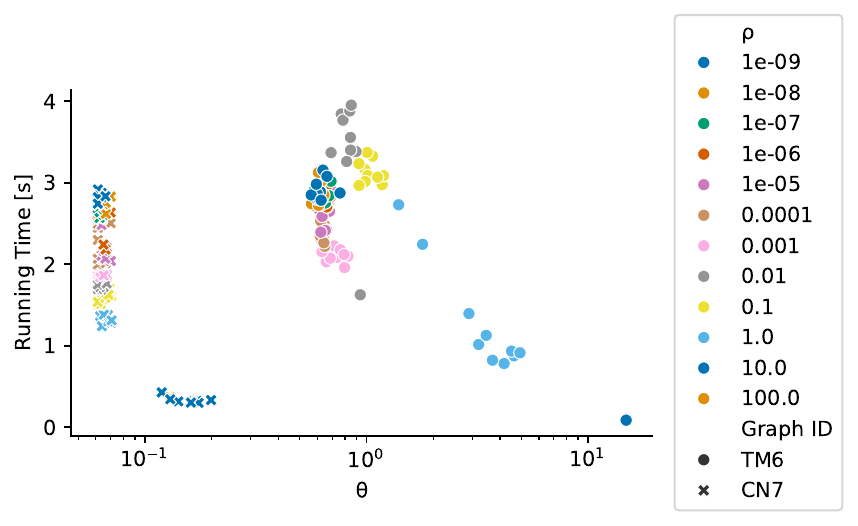}
    \caption{\mathversion{bold}Running time vs.\ normalized cut for different choices of $\rho$ on two different graphs for $k=16$. Colors denote different levels of $\rho$, while shapes indicate the graph.}
    \label{fig:potential}
\end{figure}

\paragraph{Parameter Choice for the Expander Decomposition ($\rho$, $\gamma$)}
In Figure~\ref{fig:potential}, we examine the effect of the threshold parameter $\rho$ on the quality of solutions and running time. 
If the parameter is chosen too large, especially greater than one, the entire graph will likely be certified as an expander before any cuts are made, leading to very large $\nc$. 
Furthermore, we no longer obtain any significant improvement in $\nc$ for values of $\rho$ below $10^{-4}$. 
At the same time, the running time increases inversely with $\rho$ as extra iterations are needed to converge, which is shown by the vertical arrangement of the dots in Figure~\ref{fig:potential}. 
For graph instance TM6, we also find that non-Pareto-optimal choices of $\rho$ exist, namely 0.1 and 0.01, where both the running time and the returned value are worse than $10^{-3}$ and $10^{-4}$. 
This suggests that the choice of $\rho$ is an important design decision.
On all our instances, $\rho =10^{-4}$ produced Pareto-optimal results, so we conclude that we can configure this parameter to be a constant independently of the graph's structure.

For the automatic tuning of $\gamma$, we chose a starting threshold of $0.3$, as this is around the largest value for which the expander decomposition finds structurally interesting cuts. Whenever $\gamma$ is too large, i.e., the node reduction $|G_{i+1}|/|G_i| > 0.95$, we multiply $\gamma$ by a factor of $\epsilon = 0.8$ and restart the expander decomposition.

\subsection{Comparison to \Graclus{}, \Metis{}, and \Kahip{}}\label{sec:graclus}

\autoref{fig:barplot-32} depicts the solution quality of every solver relative to that produced by \Graclus{} on all instances for $k=32$, showing both the mean and the minimum of the ten runs for \Algo{}. 
For other values of $k$, the overall picture remains the same,
see~\autoref{fig:gmean-vs-k} and~\autoref{app:figures}.
We group the graphs by type, with \autoref{fig:barplot-32} containing two plots, the upper showing the disconnected IMDB graph and email, citation, and social network graphs as well as infrastructure networks, as these are the graphs on which \Algo{} is particularly strong.

Looking at absolute values of $\nc$, we find %
that across all instances,
the geometric mean is at least \SI{70}{\percent} lower than our competitors, see \autoref{tab:graclus}.
Interestingly, when we only consider the seven graphs from~\cite{graclus}, the geometric means become 0.84 for \Algo{min} and 0.90 for \Algo{mean}, while they are 1.11, 1.19 and 1.03 for \Graclus{}, \Metis{}, and \Kahip{} respectively, which implies that \Kahip{} slightly outperforms \Graclus{} in terms of $\nc$ on their benchmark (but not \Algo{}).

\begin{figure}\Description{A lineplot plotting the geometric mean of the normalized cut value for the algorithms XCut, Graclus, METIS and KaHiP. As k increases the value of the normalized cut increases. XCut's line is at the bottom, above is Graclus with a large gap. Then a small gap above the Graclus line where both METIS and KaHiP appear close together, with METIS slightly above KaHiP.}
    \centering
    \includegraphics[width=0.6\linewidth]{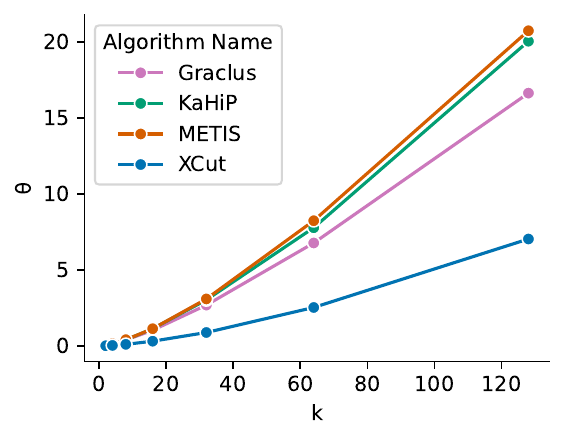}
    \caption{\mathversion{bold}Geometric mean of the cut value $\nc$ across all graphs for each $k$ for \Algo{mean}, \Graclus{}, \Metis{}, and \Kahip{}.}
    \label{fig:gmean-vs-k}
\end{figure}

One point of note is that \Algo{} does not always find small normalized cuts on the social network SN7 representing user-user interactions in the Foursquare social network.
This appears to be due to a very high-degree node that connects to approximately \SI{15}{\percent} of all vertices, leading to fast convergence of the random walk.
Due to the averaging effect of such a vertex, many nodes $v$ have almost identical values $u_v$ that vary only slightly between runs. 
Thus, when sorting by $u$-value, their order can vary greatly depending on the initial values, leading to very different candidate cuts.
This is the only graph where the minimum and mean of the ten runs of \Algo{} differ significantly ($-\SI{33}{\percent}$ for \Algo{min} and $+\SI{39}{\percent}$ for \Algo{mean} relative to Graclus).
This indicates that the theoretical algorithm sketched in \autoref{sec:expdec}, which uses multiple concurrent random walks (corresponding to more attempts to find a good cut), would likely have reduced the variance here.
This is also the only graph where the fact that we only use a single random walk impairs the quality of the result.

The graph classes on which \Algo{} does not perform as well have fairly homogeneous degree distributions and often appear grid-like when drawn. We conjecture that in these grid-like graphs, there are no good expanders (which \Algo{} is trying to find). 
Instead, sparse cuts arise mainly from the fact that the cut is balanced, i.e., the components we disconnect are all large enough, rather than
being sparsely interconnected, which
is exploited by %
the other solvers.

\begin{table}
\caption{\mathversion{bold}Cut value (\nc) of different algorithms.
The geometric mean is taken across all graphs and values of $k \in \{2, 4, 8, 16, 32, 64, 128\}$ for each graph type.
${}^*$For the overall geometric mean, instances (graph + $k$) with $\nc = 0$ were omitted.
Only \Algo{} detected such cases.
}\label{tab:graclus}
\scalebox{0.9}{
\begin{tabular}{lrrrrr}
\toprule
Type & \Algo{mean} & \Algo{min} & \Graclus{} & \Metis{} & \Kahip{} \\
\midrule
BP & \bf 0.00 & \bf 0.00 & 1.18 & 1.38 & 1.58 \\
CF & 4.61 & 4.18 & 3.44 & 3.51 & \bf 3.15 \\
CL & 0.80 & \bf 0.72 & 1.04 & 1.13 & 1.02 \\
CN & \bf 0.07 & \bf 0.07 & 1.42 & 1.72 & 1.58 \\
CS & 0.44 & \bf 0.41 & 0.51 & 0.59 & 0.57 \\
DM & 2.58 & 2.42 & \bf 2.11 & 2.12 & 2.23 \\
EM & 0.44 & \bf 0.43 & 3.44 & 3.78 & 3.80 \\
FE & 0.56 & \bf 0.53 & 0.54 & 0.54 & 0.55 \\
IF & \bf 0.00 & \bf 0.00 & 0.93 & 1.11 & 1.12 \\
NS & 0.85 & 0.78 & 0.77 & \bf 0.76 & 0.80 \\
OP & 0.71 & \bf 0.66 & 1.47 & 1.48 & 0.99 \\
RD & 13.37 & 13.37 & 12.08 & 12.01 & \bf 11.94 \\
RN & \bf 0.01 & \bf 0.01 & \bf 0.01 & \bf 0.01 & \bf 0.01 \\
SN & 0.25 & \bf 0.19 & 3.43 & 3.98 & 4.00 \\
TM & 2.04 & \bf 1.93 & 2.87 & 3.29 & 3.04 \\
US & 0.04 & \bf 0.03 & 0.06 & 0.06 & 0.05 \\
WB & \bf 0.01 & \bf 0.01 & 0.10 & 0.14 & 0.14 \\
\midrule
\bf All$^*$ & 0.25 & \bf 0.22 & 0.89 & 1.0 & 0.94 \\
\bottomrule
\end{tabular}
}
\end{table}

\begin{figure}[ht]\Description{A scatterplot depicting the ratio between average and median degree of real world graphs on the x axis and the relative value of XCut's normalized cut versus Graclus is shown on the y axis. The points show a negative trend, decreasing in value as the ratio increases.}
    \centering
    \includegraphics[width=.85\linewidth]{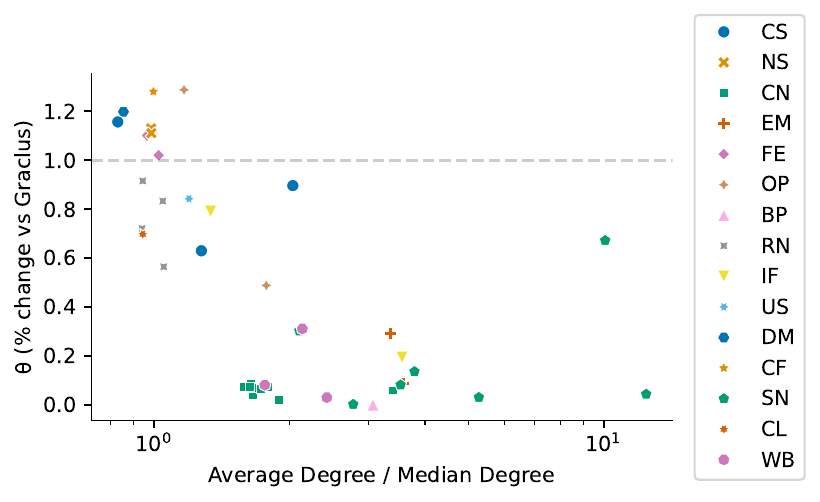}
    \caption{\mathversion{bold}Relative improvement over \Graclus{} (y-axis) vs.\ the average degree divided by the median degree (x-axis).
    Larger x-values signify that the graph exhibits a skewed, power-law-like degree distribution. Note the negative trend, except for the outlier towards the right corresponding to instance SN7.}\label{fig:discussion}
\end{figure}

\begin{figure*}[ht]\Description{The figure depicts two bar plots stacked on top of each other. On the x axis there are the IDs of different graph instances. Each instance then has four corresponding bars for the algorithms XCut, XCut minimum, KaHiP and Metis. The y axis shows the relative change in normalized cut value for k = 32 of these algorithms versus Graclus. 
The top graph contains the IMDB graph, citation networks, social networks, email networks, infrastructure networks and web graphs. The lines of XCut point down from the 0 percent change line, indicating their values are smaller, often by more than 90 percent. METIS and KaHiP's lines always point up, reaching 25-75 percent worse values.
In the bottom figure are the remaining instances, where the divergence of all solvers is within plus/minus fifty percent and no clear pattern is visible.}
    \centering
    \includegraphics[width=0.95\textwidth]{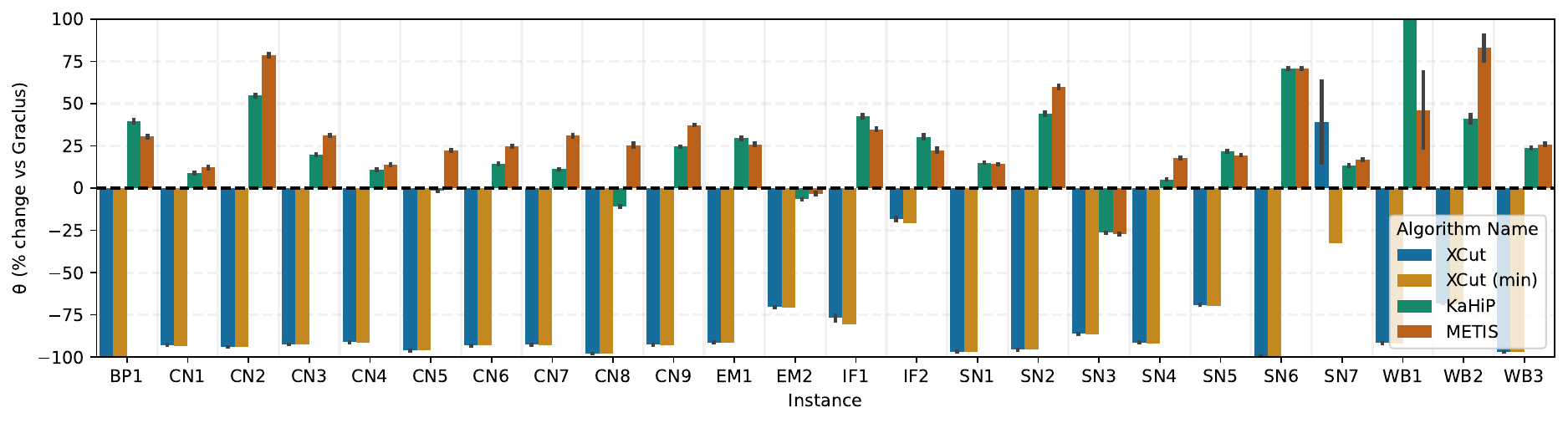}
    \\
    \includegraphics[width=0.95\textwidth]{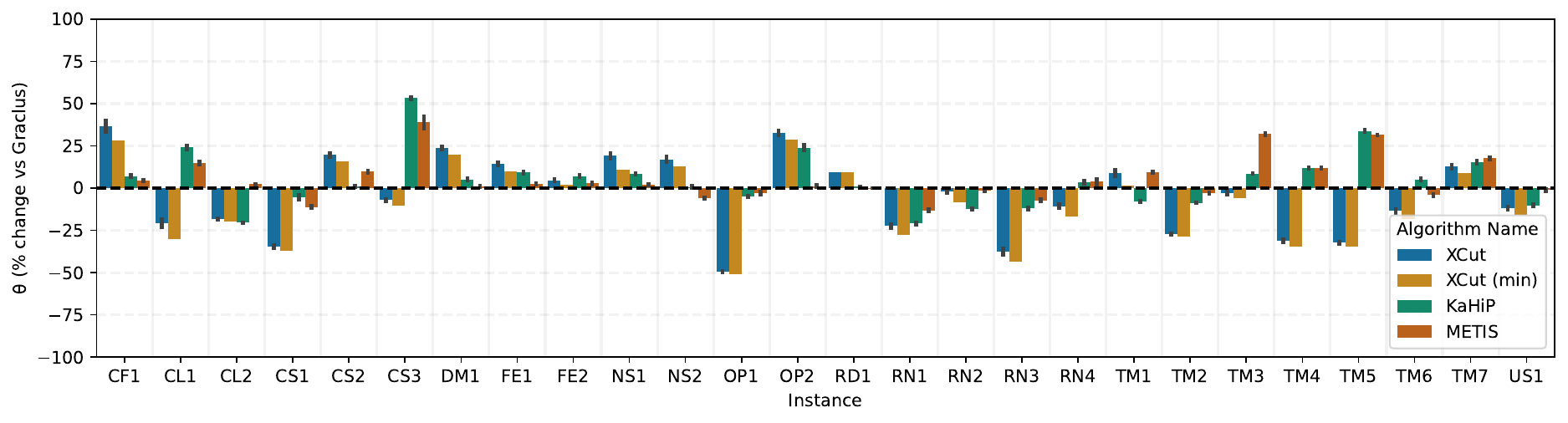}
    \caption{\mathversion{bold}Percentage deviation of the returned normalized cut value relative
      to \Graclus{} for $k=32$. This means that a value of -75\% indicates that
      the normalized cut value is 75\% lower (i.e., better). The thin black
      bars indicate the standard error across our runs. The top graph shows the
      disconnected IMDB graph (BP1), citation network instances (CN), email
      networks (EM), infrastructure graphs (IF), social networks (SN), and web
      graphs (WB), while the bottom shows the remaining instances. See the full version of the paper %
      for details. }\label{fig:barplot-32}
\end{figure*}

\subsection{Comparison to Zhao et al.}

In \autoref{tab:zhao} we compare $\nc$ for \Algo{} and \Graclus{} to the values reported by Zhao~\ea{}~\cite{zhao2018nearly} for $k = 30$.
We observe a similar but weaker pattern as in the previous section. 
On graphs arising from numerical simulation and finite element problems, \Algo{} performs worse than Zhao~\ea{}, but never by more than \SI{36}{\percent}. 
On the triangle mixture instances, \Algo{} achieves better $\nc$ on four instances, while their solver outperforms \Algo{} on two instances. 
On citation networks, clustering instances, and those based on maps (RN1 and US1), 
\Algo{} outperforms the algorithm by Zhao \ea{}, with $\nc$ being up to 2.5 times lower on US1 and CN7. 
Altogether, \Algo{} is better than Zhao~\ea{} on roughly 2/3 of the instances, while \Graclus{} does best on one instance. 
The geometric mean across all instances is 1.46 for \Algo{mean}, 1.39 for \Algo{min}, 1.64 for Zhao~\ea{}, and 3.06 for \Graclus{}.
We note that ours is \SI{11}{\percent} lower for \Algo{mean} and \SI{15}{\percent} lower for \Algo{min}, while \Graclus{}'s value is almost twice that of the other solvers in the comparison due to it producing 5--30 times greater $\nc$ on the citation network (CN) instances.

While it is difficult to draw good conclusions for the running times from the values reported in~\cite{zhao2018nearly} as no source code is available, we note that on some instances \Algo{} takes less than \SI{10}{\percent} of the reported time while producing higher-quality solutions, which might be indicative of a running time advantage of our solver.

\subsection{Running Time}

In our experiments, we found that on many graphs, the running time is spent mainly on computing the expander hierarchy, where on some instances, this step accounts for over 80\% of the running time, even for $k = 128$. 
See \autoref{fig:stacked-bars-2} for some examples. %
However, even on the largest instances in our benchmark, the absolute running time never exceeded 18 minutes.

Overall, \Algo{} is on average three times slower than \Graclus{}, \num{20.8}~times slower than \Metis{} and \num{6.7}~times slower than \Kahip{} for the mean execution time across all choices of $k$ and all instances. Interestingly, we find that \Algo{}'s running time is lower than \Graclus{}'s on several social network graphs and the triangle mixture instances while it produces a better solution. See the full version of the paper %
for detailed running times on some exemplary instances. 

Finally, recall that once our algorithm has computed a sparsifier, we can obtain a solution on the sparsifier for different values of $k$ without recomputing the sparsifier, which is a unique feature among the solvers tested here.
If we are interested in all seven values of $k$, e.g., and only count the time to compute the sparsifier once for \Algo{}, our experiments take \num{0.56}, \num{3.82}, and \num{1.24} times the running time to compute partitions across all graphs and values of $k$ using \Graclus{}, \Metis{}, and \Kahip{}, respectively. In particular, \Algo{} then is \SI{44}{\percent} \emph{faster} than \Graclus{}.
This demonstrates the utility of \Algo{} as a tool for exploratory data analysis. We could achieve further speedups by only computing a solution for $k=128$ and then choosing the initial subsets of edges for the other values of $k$, only performing the refinement.

\subsection{Discussion}

\Algo{} outperforms other software when computing normalized cuts on social, citation, email, and infrastructure networks and web graphs. 
It performs slightly worse on graphs arising from specific computational tasks, such as finite elements, circuit, or numerical simulations.
The graphs on which this behavior occurs tend to have degree distributions concentrated around the average degree, suggesting that they are not graphs with a scale-free structure.

In \autoref{fig:discussion} we plot the relationship between our improvement over \Graclus{} 
and the value of $\frac{\frac{1}{n} \sum_{v\in V}d_v}{\text{median}_{v \in V} d_v}$ for the non-synthetic graphs of our benchmark. This value measures how much the mean and median diverge due to outlier nodes with very high degrees. The graphs with power-law distributions tend to have a higher value on this measure, and we find that there appears to be a negative correlation, with one outlier due to instance SN7. %

\balance

\section{Conclusion}
In this work we introduced \Algo{}, a new algorithm for solving the normalized cut problem. It is based on a novel expander decomposition algorithm and to the best of our knowledge, it is the first practical application of the expander hierarchy.
\Algo{} clearly outperforms other solvers in the experimental study on social, citation, email, and infrastructure networks and web graphs, and also in the geometric mean over all instances.
It scales to instances with tens of millions of edges and can produce solutions for multiple numbers of clusters $k$ with little overhead by comparison.
We are confident that with further optimization and the use of parallelism it will be possible to scale our algorithm to even larger graphs, while further improving the solution quality, especially since computing the expander decomposition appears to be highly parallelizable.

We also believe that the expander hierarchy and our expander decomposition can be applied to other graph cut problems in the future, as the tree flow sparsifiers approximate \emph{all} cuts in the graph, and there are theoretical results that suggest this might be the case.
\Algo{} is open source software and its code is freely available on GitLab~\cite{XCutSource}.

\begin{acks}
\erclogowrapped{5\baselineskip} 
Monika Henzinger:  This project has received funding from the European Research Council (ERC) under the European Union's Horizon 2020 research and innovation programme (Grant agreement No. 101019564) and the Austrian Science Fund (FWF) grant DOI 10.55776/Z422, grant DOI 10.55776/I5982, and grant DOI 10.55776/P33775 with additional funding from the netidee SCIENCE Stiftung, 2020–2024.

Harald Räcke, Robin Münk: This project has received funding from the Deutsche Forschungsgemeinschaft (DFG, German Research Foundation) – 498605858 and 470029389.

\end{acks}

\clearpage

\apptocmd{\thebibliography}{\hbadness 4000\emergencystretch1em\relax}{}{}
\bibliographystyle{ACM-Reference-Format}
\bibliography{main}


\begin{thebibliography}{43}


\ifx \showCODEN    \undefined \def \showCODEN     #1{\unskip}     \fi
\ifx \showDOI      \undefined \def \showDOI       #1{#1}\fi
\ifx \showISBNx    \undefined \def \showISBNx     #1{\unskip}     \fi
\ifx \showISBNxiii \undefined \def \showISBNxiii  #1{\unskip}     \fi
\ifx \showISSN     \undefined \def \showISSN      #1{\unskip}     \fi
\ifx \showLCCN     \undefined \def \showLCCN      #1{\unskip}     \fi
\ifx \shownote     \undefined \def \shownote      #1{#1}          \fi
\ifx \showarticletitle \undefined \def \showarticletitle #1{#1}   \fi
\ifx \showURL      \undefined \def \showURL       {\relax}        \fi
\providecommand\bibfield[2]{#2}
\providecommand\bibinfo[2]{#2}
\providecommand\natexlab[1]{#1}
\providecommand\showeprint[2][]{arXiv:#2}

\bibitem[\protect\citeauthoryear{Arvestad}{Arvestad}{2022}]%
        {arvestad2022near}
\bibfield{author}{\bibinfo{person}{Isaac Arvestad}.}
  \bibinfo{year}{2022}\natexlab{}.
\newblock \bibinfo{title}{Near-linear time expander decomposition in practice}.
\newblock
\newblock


\bibitem[\protect\citeauthoryear{Auer and Bisseling}{Auer and
  Bisseling}{2012}]%
        {auer2012graph}
\bibfield{author}{\bibinfo{person}{Bas~Fagginger Auer} {and}
  \bibinfo{person}{Rob~H Bisseling}.} \bibinfo{year}{2012}\natexlab{}.
\newblock \showarticletitle{Graph coarsening and clustering on the GPU.}
\newblock \bibinfo{journal}{\emph{Graph Partitioning and Graph Clustering}}
  \bibinfo{volume}{588}, \bibinfo{number}{223} (\bibinfo{year}{2012}),
  \bibinfo{pages}{2}.
\newblock


\bibitem[\protect\citeauthoryear{Bichot}{Bichot}{2010}]%
        {DBLP:journals/jmma/Bichot10}
\bibfield{author}{\bibinfo{person}{Charles{-}Edmond Bichot}.}
  \bibinfo{year}{2010}\natexlab{}.
\newblock \showarticletitle{Co-clustering Documents and Words by Minimizing the
  Normalized Cut Objective Function}.
\newblock \bibinfo{journal}{\emph{J. Math. Model. Algorithms}}
  \bibinfo{volume}{9}, \bibinfo{number}{2} (\bibinfo{year}{2010}),
  \bibinfo{pages}{131--147}.
\newblock
\urldef\tempurl%
\url{https://doi.org/10.1007/S10852-010-9126-0}
\showDOI{\tempurl}


\bibitem[\protect\citeauthoryear{Cai, Shao, He, Yan, and Han}{Cai
  et~al\mbox{.}}{2005}]%
        {DBLP:conf/kdd/CaiSHYH05}
\bibfield{author}{\bibinfo{person}{Deng Cai}, \bibinfo{person}{Zheng Shao},
  \bibinfo{person}{Xiaofei He}, \bibinfo{person}{Xifeng Yan}, {and}
  \bibinfo{person}{Jiawei Han}.} \bibinfo{year}{2005}\natexlab{}.
\newblock \showarticletitle{Mining hidden community in heterogeneous social
  networks}. In \bibinfo{booktitle}{\emph{Proceedings of the 3rd international
  workshop on Link discovery, LinkKDD 2005, Chicago, Illinois, USA, August
  21-25, 2005}}, \bibfield{editor}{\bibinfo{person}{Jafar Adibi},
  \bibinfo{person}{Marko Grobelnik}, \bibinfo{person}{Dunja Mladenic}, {and}
  \bibinfo{person}{Patrick Pantel}} (Eds.). \bibinfo{publisher}{{ACM}},
  \bibinfo{pages}{58--65}.
\newblock
\urldef\tempurl%
\url{https://doi.org/10.1145/1134271.1134280}
\showDOI{\tempurl}


\bibitem[\protect\citeauthoryear{Chang, Pettie, Saranurak, and Zhang}{Chang
  et~al\mbox{.}}{2021}]%
        {DBLP:journals/jacm/ChangPSZ21}
\bibfield{author}{\bibinfo{person}{Yi{-}Jun Chang}, \bibinfo{person}{Seth
  Pettie}, \bibinfo{person}{Thatchaphol Saranurak}, {and}
  \bibinfo{person}{Hengjie Zhang}.} \bibinfo{year}{2021}\natexlab{}.
\newblock \showarticletitle{Near-optimal Distributed Triangle Enumeration via
  Expander Decompositions}.
\newblock \bibinfo{journal}{\emph{J. {ACM}}} \bibinfo{volume}{68},
  \bibinfo{number}{3} (\bibinfo{year}{2021}), \bibinfo{pages}{21:1--21:36}.
\newblock
\urldef\tempurl%
\url{https://doi.org/10.1145/3446330}
\showDOI{\tempurl}


\bibitem[\protect\citeauthoryear{Chen, Kyng, Liu, Peng, Gutenberg, and
  Sachdeva}{Chen et~al\mbox{.}}{2022}]%
        {DBLP:conf/focs/ChenKLPGS22}
\bibfield{author}{\bibinfo{person}{Li Chen}, \bibinfo{person}{Rasmus Kyng},
  \bibinfo{person}{Yang~P. Liu}, \bibinfo{person}{Richard Peng},
  \bibinfo{person}{Maximilian~Probst Gutenberg}, {and} \bibinfo{person}{Sushant
  Sachdeva}.} \bibinfo{year}{2022}\natexlab{}.
\newblock \showarticletitle{Maximum Flow and Minimum-Cost Flow in Almost-Linear
  Time}. In \bibinfo{booktitle}{\emph{63rd {IEEE} Annual Symposium on
  Foundations of Computer Science, {FOCS} 2022, Denver, CO, USA, October 31 -
  November 3, 2022}}. \bibinfo{publisher}{{IEEE}}, \bibinfo{pages}{612--623}.
\newblock
\urldef\tempurl%
\url{https://doi.org/10.1109/FOCS54457.2022.00064}
\showDOI{\tempurl}


\bibitem[\protect\citeauthoryear{Chen, Nie, Huang, and Yang}{Chen
  et~al\mbox{.}}{2017}]%
        {DBLP:conf/ijcai/0004NHY17}
\bibfield{author}{\bibinfo{person}{Xiaojun Chen}, \bibinfo{person}{Feiping
  Nie}, \bibinfo{person}{Joshua~Zhexue Huang}, {and} \bibinfo{person}{Min
  Yang}.} \bibinfo{year}{2017}\natexlab{}.
\newblock \showarticletitle{Scalable Normalized Cut with Improved Spectral
  Rotation}. In \bibinfo{booktitle}{\emph{Proceedings of the Twenty-Sixth
  International Joint Conference on Artificial Intelligence, {IJCAI} 2017,
  Melbourne, Australia, August 19-25, 2017}},
  \bibfield{editor}{\bibinfo{person}{Carles Sierra}} (Ed.).
  \bibinfo{publisher}{ijcai.org}, \bibinfo{pages}{1518--1524}.
\newblock
\urldef\tempurl%
\url{https://doi.org/10.24963/IJCAI.2017/210}
\showDOI{\tempurl}


\bibitem[\protect\citeauthoryear{Daneshgar and Javadi}{Daneshgar and
  Javadi}{2012}]%
        {daneshgar2012complexity}
\bibfield{author}{\bibinfo{person}{Amir Daneshgar} {and} \bibinfo{person}{Ramin
  Javadi}.} \bibinfo{year}{2012}\natexlab{}.
\newblock \showarticletitle{On the complexity of isoperimetric problems on
  trees}.
\newblock \bibinfo{journal}{\emph{Discret. Appl. Math.}} \bibinfo{volume}{160},
  \bibinfo{number}{1-2} (\bibinfo{year}{2012}), \bibinfo{pages}{116--131}.
\newblock
\urldef\tempurl%
\url{https://doi.org/10.1016/J.DAM.2011.08.015}
\showDOI{\tempurl}


\bibitem[\protect\citeauthoryear{Davis and Hu}{Davis and Hu}{2011}]%
        {suitesparse}
\bibfield{author}{\bibinfo{person}{Timothy~A Davis} {and}
  \bibinfo{person}{Yifan Hu}.} \bibinfo{year}{2011}\natexlab{}.
\newblock \showarticletitle{The University of Florida sparse matrix
  collection}.
\newblock \bibinfo{journal}{\emph{ACM Transactions on Mathematical Software
  (TOMS)}} \bibinfo{volume}{38}, \bibinfo{number}{1} (\bibinfo{year}{2011}),
  \bibinfo{pages}{1--25}.
\newblock


\bibitem[\protect\citeauthoryear{Dhillon, Guan, and Kulis}{Dhillon
  et~al\mbox{.}}{2007}]%
        {graclus}
\bibfield{author}{\bibinfo{person}{Inderjit~S Dhillon},
  \bibinfo{person}{Yuqiang Guan}, {and} \bibinfo{person}{Brian Kulis}.}
  \bibinfo{year}{2007}\natexlab{}.
\newblock \showarticletitle{Weighted graph cuts without eigenvectors a
  multilevel approach}.
\newblock \bibinfo{journal}{\emph{IEEE transactions on pattern analysis and
  machine intelligence}} \bibinfo{volume}{29}, \bibinfo{number}{11}
  (\bibinfo{year}{2007}), \bibinfo{pages}{1944--1957}.
\newblock


\bibitem[\protect\citeauthoryear{Goranci, R{\"a}cke, Saranurak, and
  Tan}{Goranci et~al\mbox{.}}{2021}]%
        {expanderHierarchy}
\bibfield{author}{\bibinfo{person}{Gramoz Goranci}, \bibinfo{person}{Harald
  R{\"a}cke}, \bibinfo{person}{Thatchaphol Saranurak}, {and}
  \bibinfo{person}{Zihan Tan}.} \bibinfo{year}{2021}\natexlab{}.
\newblock \showarticletitle{The expander hierarchy and its applications to
  dynamic graph algorithms}. In \bibinfo{booktitle}{\emph{Proceedings of the
  2021 ACM-SIAM Symposium on Discrete Algorithms (SODA)}}. SIAM,
  \bibinfo{pages}{2212--2228}.
\newblock


\bibitem[\protect\citeauthoryear{Haeupler, R{\"a}cke, and Ghaffari}{Haeupler
  et~al\mbox{.}}{2022}]%
        {haeupler2022hop}
\bibfield{author}{\bibinfo{person}{Bernhard Haeupler}, \bibinfo{person}{Harald
  R{\"a}cke}, {and} \bibinfo{person}{Mohsen Ghaffari}.}
  \bibinfo{year}{2022}\natexlab{}.
\newblock \showarticletitle{Hop-constrained expander decompositions, oblivious
  routing, and distributed universal optimality}. In
  \bibinfo{booktitle}{\emph{Proceedings of the 54th Annual ACM SIGACT Symposium
  on Theory of Computing}}. \bibinfo{pages}{1325--1338}.
\newblock


\bibitem[\protect\citeauthoryear{Hanauer, Henzinger, M\"unk, R\"acke, and
  V\"otsch}{Hanauer et~al\mbox{.}}{2024a}]%
        {kddofficial}
\bibfield{author}{\bibinfo{person}{Kathrin Hanauer}, \bibinfo{person}{Monika
  Henzinger}, \bibinfo{person}{Robin M\"unk}, \bibinfo{person}{Harald R\"acke},
  {and} \bibinfo{person}{Maximilian V\"otsch}.}
  \bibinfo{year}{2024}\natexlab{a}.
\newblock \showarticletitle{Expander Hierarchies for Normalized Cuts on
  Graphs}. In \bibinfo{booktitle}{\emph{Proceedings of the 30th {ACM} {SIGKDD}
  Conference on Knowledge Discovery and Data Mining, {KDD} 2023, August
  25–29, 2024, Barcelona, Spain}}. \bibinfo{publisher}{{ACM}}.
\newblock
\urldef\tempurl%
\url{https://doi.org/10.1145/3637528.3671978}
\showDOI{\tempurl}
\newblock
\shownote{To appear}.


\bibitem[\protect\citeauthoryear{Hanauer, Henzinger, M\"unk, R\"acke, and
  V\"otsch}{Hanauer et~al\mbox{.}}{2024b}]%
        {XCutSource}
\bibfield{author}{\bibinfo{person}{Kathrin Hanauer}, \bibinfo{person}{Monika
  Henzinger}, \bibinfo{person}{Robin M\"unk}, \bibinfo{person}{Harald R\"acke},
  {and} \bibinfo{person}{Maximilian V\"otsch}.}
  \bibinfo{year}{2024}\natexlab{b}.
\newblock \bibinfo{title}{XCut/XCut v1.0.0}.
\newblock
  \bibinfo{howpublished}{\url{https://doi.org/10.5281/zenodo.12108189}}.
\newblock
\urldef\tempurl%
\url{https://doi.org/10.5281/zenodo.12108189}
\showDOI{\tempurl}


\bibitem[\protect\citeauthoryear{Hendrickson and Leland}{Hendrickson and
  Leland}{1995}]%
        {chaco}
\bibfield{author}{\bibinfo{person}{Bruce Hendrickson} {and}
  \bibinfo{person}{Robert~W. Leland}.} \bibinfo{year}{1995}\natexlab{}.
\newblock \showarticletitle{A Multi-Level Algorithm For Partitioning Graphs}.
  In \bibinfo{booktitle}{\emph{Proceedings Supercomputing '95, San Diego, CA,
  USA, December 4-8, 1995}}, \bibfield{editor}{\bibinfo{person}{Sidney Karin}}
  (Ed.). \bibinfo{publisher}{{ACM}}, \bibinfo{pages}{28}.
\newblock
\urldef\tempurl%
\url{https://doi.org/10.1145/224170.224228}
\showDOI{\tempurl}


\bibitem[\protect\citeauthoryear{Henzinger, Neumann, R{\"{a}}cke, and
  Schmid}{Henzinger et~al\mbox{.}}{2023}]%
        {dpSTACS}
\bibfield{author}{\bibinfo{person}{Monika Henzinger}, \bibinfo{person}{Stefan
  Neumann}, \bibinfo{person}{Harald R{\"{a}}cke}, {and} \bibinfo{person}{Stefan
  Schmid}.} \bibinfo{year}{2023}\natexlab{}.
\newblock \showarticletitle{Dynamic Maintenance of Monotone Dynamic Programs
  and Applications}. In \bibinfo{booktitle}{\emph{40th International Symposium
  on Theoretical Aspects of Computer Science, {STACS} 2023, March 7-9, 2023,
  Hamburg, Germany}} \emph{(\bibinfo{series}{LIPIcs},
  Vol.~\bibinfo{volume}{254})}, \bibfield{editor}{\bibinfo{person}{Petra
  Berenbrink}, \bibinfo{person}{Patricia Bouyer}, \bibinfo{person}{Anuj Dawar},
  {and} \bibinfo{person}{Mamadou~Moustapha Kant{\'{e}}}} (Eds.).
  \bibinfo{publisher}{Schloss Dagstuhl - Leibniz-Zentrum f{\"{u}}r Informatik},
  \bibinfo{pages}{36:1--36:16}.
\newblock
\urldef\tempurl%
\url{https://doi.org/10.4230/LIPICS.STACS.2023.36}
\showDOI{\tempurl}


\bibitem[\protect\citeauthoryear{Hua, Kyng, Gutenberg, and Wu}{Hua
  et~al\mbox{.}}{2023}]%
        {hua2023maintaining}
\bibfield{author}{\bibinfo{person}{Yiding Hua}, \bibinfo{person}{Rasmus Kyng},
  \bibinfo{person}{Maximilian~Probst Gutenberg}, {and} \bibinfo{person}{Zihang
  Wu}.} \bibinfo{year}{2023}\natexlab{}.
\newblock \showarticletitle{Maintaining expander decompositions via sparse
  cuts}. In \bibinfo{booktitle}{\emph{Proceedings of the 2023 Annual ACM-SIAM
  Symposium on Discrete Algorithms (SODA)}}. SIAM, \bibinfo{pages}{48--69}.
\newblock


\bibitem[\protect\citeauthoryear{Karypis and Kumar}{Karypis and Kumar}{1998}]%
        {metis}
\bibfield{author}{\bibinfo{person}{George Karypis} {and} \bibinfo{person}{Vipin
  Kumar}.} \bibinfo{year}{1998}\natexlab{}.
\newblock \showarticletitle{A fast and high quality multilevel scheme for
  partitioning irregular graphs}.
\newblock \bibinfo{journal}{\emph{SIAM Journal on scientific Computing}}
  \bibinfo{volume}{20}, \bibinfo{number}{1} (\bibinfo{year}{1998}),
  \bibinfo{pages}{359--392}.
\newblock


\bibitem[\protect\citeauthoryear{Kelner, Lee, Orecchia, and Sidford}{Kelner
  et~al\mbox{.}}{2014}]%
        {DBLP:conf/soda/KelnerLOS14}
\bibfield{author}{\bibinfo{person}{Jonathan~A. Kelner},
  \bibinfo{person}{Yin~Tat Lee}, \bibinfo{person}{Lorenzo Orecchia}, {and}
  \bibinfo{person}{Aaron Sidford}.} \bibinfo{year}{2014}\natexlab{}.
\newblock \showarticletitle{An Almost-Linear-Time Algorithm for Approximate Max
  Flow in Undirected Graphs, and its Multicommodity Generalizations}. In
  \bibinfo{booktitle}{\emph{Proceedings of the Twenty-Fifth Annual {ACM-SIAM}
  Symposium on Discrete Algorithms, {SODA} 2014, Portland, Oregon, USA, January
  5-7, 2014}}, \bibfield{editor}{\bibinfo{person}{Chandra Chekuri}} (Ed.).
  \bibinfo{publisher}{{SIAM}}, \bibinfo{pages}{217--226}.
\newblock
\urldef\tempurl%
\url{https://doi.org/10.1137/1.9781611973402.16}
\showDOI{\tempurl}


\bibitem[\protect\citeauthoryear{Laurent and Massart}{Laurent and
  Massart}{2000}]%
        {LM2000}
\bibfield{author}{\bibinfo{person}{B. Laurent} {and} \bibinfo{person}{P.
  Massart}.} \bibinfo{year}{2000}\natexlab{}.
\newblock \showarticletitle{{Adaptive estimation of a quadratic functional by
  model selection}}.
\newblock \bibinfo{journal}{\emph{The Annals of Statistics}}
  \bibinfo{volume}{28}, \bibinfo{number}{5} (\bibinfo{year}{2000}),
  \bibinfo{pages}{1302 -- 1338}.
\newblock
\urldef\tempurl%
\url{https://doi.org/10.1214/aos/1015957395}
\showDOI{\tempurl}


\bibitem[\protect\citeauthoryear{Li}{Li}{2021}]%
        {li2021deterministic}
\bibfield{author}{\bibinfo{person}{Jason Li}.} \bibinfo{year}{2021}\natexlab{}.
\newblock \showarticletitle{Deterministic mincut in almost-linear time}. In
  \bibinfo{booktitle}{\emph{Proceedings of the 53rd Annual ACM SIGACT Symposium
  on Theory of Computing}}. \bibinfo{pages}{384--395}.
\newblock


\bibitem[\protect\citeauthoryear{Li, Zhang, Han, Rong, Cheng, and Huang}{Li
  et~al\mbox{.}}{2020}]%
        {DBLP:conf/www/LiZHRCH20}
\bibfield{author}{\bibinfo{person}{Jia Li}, \bibinfo{person}{Honglei Zhang},
  \bibinfo{person}{Zhichao Han}, \bibinfo{person}{Yu Rong},
  \bibinfo{person}{Hong Cheng}, {and} \bibinfo{person}{Junzhou Huang}.}
  \bibinfo{year}{2020}\natexlab{}.
\newblock \showarticletitle{Adversarial Attack on Community Detection by Hiding
  Individuals}. In \bibinfo{booktitle}{\emph{{WWW} '20: The Web Conference
  2020, Taipei, Taiwan, April 20-24, 2020}},
  \bibfield{editor}{\bibinfo{person}{Yennun Huang}, \bibinfo{person}{Irwin
  King}, \bibinfo{person}{Tie{-}Yan Liu}, {and} \bibinfo{person}{Maarten van
  Steen}} (Eds.). \bibinfo{publisher}{{ACM} / {IW3C2}},
  \bibinfo{pages}{917--927}.
\newblock
\urldef\tempurl%
\url{https://doi.org/10.1145/3366423.3380171}
\showDOI{\tempurl}


\bibitem[\protect\citeauthoryear{Mohar}{Mohar}{1989}]%
        {mohar1989isoperimetric}
\bibfield{author}{\bibinfo{person}{Bojan Mohar}.}
  \bibinfo{year}{1989}\natexlab{}.
\newblock \showarticletitle{Isoperimetric numbers of graphs}.
\newblock \bibinfo{journal}{\emph{J. Comb. Theory, Ser. {B}}}
  \bibinfo{volume}{47}, \bibinfo{number}{3} (\bibinfo{year}{1989}),
  \bibinfo{pages}{274--291}.
\newblock
\urldef\tempurl%
\url{https://doi.org/10.1016/0095-8956(89)90029-4}
\showDOI{\tempurl}


\bibitem[\protect\citeauthoryear{Nanongkai and Saranurak}{Nanongkai and
  Saranurak}{2017}]%
        {nanongkai2017dynamic}
\bibfield{author}{\bibinfo{person}{Danupon Nanongkai} {and}
  \bibinfo{person}{Thatchaphol Saranurak}.} \bibinfo{year}{2017}\natexlab{}.
\newblock \showarticletitle{Dynamic spanning forest with worst-case update
  time: adaptive, las vegas, and o (n1/2-$\varepsilon$)-time}. In
  \bibinfo{booktitle}{\emph{Proceedings of the 49th Annual ACM SIGACT Symposium
  on Theory of Computing}}. \bibinfo{pages}{1122--1129}.
\newblock


\bibitem[\protect\citeauthoryear{Nanongkai, Saranurak, and
  Wulff-Nilsen}{Nanongkai et~al\mbox{.}}{2017}]%
        {nanongkai2017dynamicB}
\bibfield{author}{\bibinfo{person}{Danupon Nanongkai},
  \bibinfo{person}{Thatchaphol Saranurak}, {and} \bibinfo{person}{Christian
  Wulff-Nilsen}.} \bibinfo{year}{2017}\natexlab{}.
\newblock \showarticletitle{Dynamic minimum spanning forest with subpolynomial
  worst-case update time}. In \bibinfo{booktitle}{\emph{2017 IEEE 58th Annual
  Symposium on Foundations of Computer Science (FOCS)}}. IEEE,
  \bibinfo{pages}{950--961}.
\newblock


\bibitem[\protect\citeauthoryear{Naumov and Moon}{Naumov and Moon}{2016}]%
        {nvidia}
\bibfield{author}{\bibinfo{person}{Maxim Naumov} {and} \bibinfo{person}{Timothy
  Moon}.} \bibinfo{year}{2016}\natexlab{}.
\newblock \showarticletitle{Parallel spectral graph partitioning}.
\newblock \bibinfo{journal}{\emph{NVIDIA, Santa Clara, CA, USA, Tech. Rep.,
  NVR-2016-001}} (\bibinfo{year}{2016}).
\newblock


\bibitem[\protect\citeauthoryear{Ng, Jordan, and Weiss}{Ng
  et~al\mbox{.}}{2001}]%
        {DBLP:conf/nips/NgJW01}
\bibfield{author}{\bibinfo{person}{Andrew~Y. Ng}, \bibinfo{person}{Michael~I.
  Jordan}, {and} \bibinfo{person}{Yair Weiss}.}
  \bibinfo{year}{2001}\natexlab{}.
\newblock \showarticletitle{On Spectral Clustering: Analysis and an algorithm}.
  In \bibinfo{booktitle}{\emph{Advances in Neural Information Processing
  Systems 14 [Neural Information Processing Systems: Natural and Synthetic,
  {NIPS} 2001, December 3-8, 2001, Vancouver, British Columbia, Canada]}},
  \bibfield{editor}{\bibinfo{person}{Thomas~G. Dietterich},
  \bibinfo{person}{Suzanna Becker}, {and} \bibinfo{person}{Zoubin Ghahramani}}
  (Eds.). \bibinfo{publisher}{{MIT} Press}, \bibinfo{pages}{849--856}.
\newblock
\urldef\tempurl%
\url{https://proceedings.neurips.cc/paper/2001/hash/801272ee79cfde7fa5960571fee36b9b-Abstract.html}
\showURL{%
\tempurl}


\bibitem[\protect\citeauthoryear{Nie, Lu, Wu, Wang, and Li}{Nie
  et~al\mbox{.}}{2024}]%
        {DBLP:journals/pami/NieLWWL24}
\bibfield{author}{\bibinfo{person}{Feiping Nie}, \bibinfo{person}{Jitao Lu},
  \bibinfo{person}{Danyang Wu}, \bibinfo{person}{Rong Wang}, {and}
  \bibinfo{person}{Xuelong Li}.} \bibinfo{year}{2024}\natexlab{}.
\newblock \showarticletitle{A Novel Normalized-Cut Solver With Nearest Neighbor
  Hierarchical Initialization}.
\newblock \bibinfo{journal}{\emph{{IEEE} Trans. Pattern Anal. Mach. Intell.}}
  \bibinfo{volume}{46}, \bibinfo{number}{1} (\bibinfo{year}{2024}),
  \bibinfo{pages}{659--666}.
\newblock
\urldef\tempurl%
\url{https://doi.org/10.1109/TPAMI.2023.3279394}
\showDOI{\tempurl}


\bibitem[\protect\citeauthoryear{Osipov and Sanders}{Osipov and
  Sanders}{2010}]%
        {osipov2010n}
\bibfield{author}{\bibinfo{person}{Vitaly Osipov} {and} \bibinfo{person}{Peter
  Sanders}.} \bibinfo{year}{2010}\natexlab{}.
\newblock \showarticletitle{n-Level graph partitioning}. In
  \bibinfo{booktitle}{\emph{Algorithms--ESA 2010: 18th Annual European
  Symposium, Liverpool, UK, September 6-8, 2010. Proceedings, Part I 18}}.
  Springer, \bibinfo{pages}{278--289}.
\newblock


\bibitem[\protect\citeauthoryear{Peng, Sun, and Zanetti}{Peng
  et~al\mbox{.}}{2017}]%
        {DBLP:journals/siamcomp/Peng0Z17}
\bibfield{author}{\bibinfo{person}{Richard Peng}, \bibinfo{person}{He Sun},
  {and} \bibinfo{person}{Luca Zanetti}.} \bibinfo{year}{2017}\natexlab{}.
\newblock \showarticletitle{Partitioning Well-Clustered Graphs: Spectral
  Clustering Works!}
\newblock \bibinfo{journal}{\emph{{SIAM} J. Comput.}} \bibinfo{volume}{46},
  \bibinfo{number}{2} (\bibinfo{year}{2017}), \bibinfo{pages}{710--743}.
\newblock
\urldef\tempurl%
\url{https://doi.org/10.1137/15M1047209}
\showDOI{\tempurl}


\bibitem[\protect\citeauthoryear{Rossi and Ahmed}{Rossi and Ahmed}{2015}]%
        {nr}
\bibfield{author}{\bibinfo{person}{Ryan~A. Rossi} {and}
  \bibinfo{person}{Nesreen~K. Ahmed}.} \bibinfo{year}{2015}\natexlab{}.
\newblock \showarticletitle{The Network Data Repository with Interactive Graph
  Analytics and Visualization}. In \bibinfo{booktitle}{\emph{AAAI}}.
\newblock
\urldef\tempurl%
\url{https://networkrepository.com}
\showURL{%
\tempurl}


\bibitem[\protect\citeauthoryear{Sahoo and Das}{Sahoo and Das}{2023}]%
        {sahoo2023brain}
\bibfield{author}{\bibinfo{person}{Tapasmini Sahoo} {and}
  \bibinfo{person}{Kunal~Kumar Das}.} \bibinfo{year}{2023}\natexlab{}.
\newblock \showarticletitle{Brain Tumor Localization Using N-Cut.}
\newblock \bibinfo{journal}{\emph{International Journal of Online \& Biomedical
  Engineering}} \bibinfo{volume}{19}, \bibinfo{number}{15}
  (\bibinfo{year}{2023}).
\newblock


\bibitem[\protect\citeauthoryear{Sanders and Schulz}{Sanders and
  Schulz}{2012}]%
        {kahip}
\bibfield{author}{\bibinfo{person}{Peter Sanders} {and}
  \bibinfo{person}{Christian Schulz}.} \bibinfo{year}{2012}\natexlab{}.
\newblock \showarticletitle{High quality graph partitioning.}
\newblock \bibinfo{journal}{\emph{Graph Partitioning and Graph Clustering}}
  \bibinfo{volume}{588}, \bibinfo{number}{1} (\bibinfo{year}{2012}),
  \bibinfo{pages}{1--17}.
\newblock


\bibitem[\protect\citeauthoryear{Saranurak and Wang}{Saranurak and
  Wang}{2019}]%
        {saranurakwang19}
\bibfield{author}{\bibinfo{person}{Thatchaphol Saranurak} {and}
  \bibinfo{person}{Di Wang}.} \bibinfo{year}{2019}\natexlab{}.
\newblock \showarticletitle{Expander Decomposition and Pruning: Faster,
  Stronger, and Simpler}. In \bibinfo{booktitle}{\emph{Proceedings of the
  Thirtieth Annual {ACM-SIAM} Symposium on Discrete Algorithms, {SODA} 2019,
  San Diego, California, USA, January 6-9, 2019}},
  \bibfield{editor}{\bibinfo{person}{Timothy~M. Chan}} (Ed.).
  \bibinfo{publisher}{{SIAM}}, \bibinfo{pages}{2616--2635}.
\newblock
\urldef\tempurl%
\url{https://doi.org/10.1137/1.9781611975482.162}
\showDOI{\tempurl}


\bibitem[\protect\citeauthoryear{Shi and Malik}{Shi and Malik}{2000}]%
        {DBLP:journals/pami/ShiM00}
\bibfield{author}{\bibinfo{person}{Jianbo Shi} {and} \bibinfo{person}{Jitendra
  Malik}.} \bibinfo{year}{2000}\natexlab{}.
\newblock \showarticletitle{Normalized Cuts and Image Segmentation}.
\newblock \bibinfo{journal}{\emph{{IEEE} Trans. Pattern Anal. Mach. Intell.}}
  \bibinfo{volume}{22}, \bibinfo{number}{8} (\bibinfo{year}{2000}),
  \bibinfo{pages}{888--905}.
\newblock
\urldef\tempurl%
\url{https://doi.org/10.1109/34.868688}
\showDOI{\tempurl}


\bibitem[\protect\citeauthoryear{Spielman and Teng}{Spielman and Teng}{2004}]%
        {spielman2004nearly}
\bibfield{author}{\bibinfo{person}{Daniel~A Spielman} {and}
  \bibinfo{person}{Shang-Hua Teng}.} \bibinfo{year}{2004}\natexlab{}.
\newblock \showarticletitle{Nearly-linear time algorithms for graph
  partitioning, graph sparsification, and solving linear systems}. In
  \bibinfo{booktitle}{\emph{Proceedings of the thirty-sixth annual ACM
  symposium on Theory of computing}}. \bibinfo{pages}{81--90}.
\newblock


\bibitem[\protect\citeauthoryear{Walshaw, Cross, and Everett}{Walshaw
  et~al\mbox{.}}{1995}]%
        {jostle}
\bibfield{author}{\bibinfo{person}{Chris Walshaw}, \bibinfo{person}{Mark
  Cross}, {and} \bibinfo{person}{Martin~G. Everett}.}
  \bibinfo{year}{1995}\natexlab{}.
\newblock \showarticletitle{A Localized Algorithm for Optimizing Unstructured
  Mesh Partitions}.
\newblock \bibinfo{journal}{\emph{Int. J. High Perform. Comput. Appl.}}
  \bibinfo{volume}{9}, \bibinfo{number}{4} (\bibinfo{year}{1995}),
  \bibinfo{pages}{280--295}.
\newblock
\urldef\tempurl%
\url{https://doi.org/10.1177/109434209500900403}
\showDOI{\tempurl}


\bibitem[\protect\citeauthoryear{Wang, Shen, Yuan, Du, Li, Hu, Crowley, and
  Vaufreydaz}{Wang et~al\mbox{.}}{2023}]%
        {tokencut}
\bibfield{author}{\bibinfo{person}{Yangtao Wang}, \bibinfo{person}{Xi Shen},
  \bibinfo{person}{Yuan Yuan}, \bibinfo{person}{Yuming Du},
  \bibinfo{person}{Maomao Li}, \bibinfo{person}{Shell~Xu Hu},
  \bibinfo{person}{James~L. Crowley}, {and} \bibinfo{person}{Dominique
  Vaufreydaz}.} \bibinfo{year}{2023}\natexlab{}.
\newblock \showarticletitle{TokenCut: Segmenting Objects in Images and Videos
  With Self-Supervised Transformer and Normalized Cut}.
\newblock \bibinfo{journal}{\emph{IEEE Transactions on Pattern Analysis and
  Machine Intelligence}} \bibinfo{volume}{45}, \bibinfo{number}{12}
  (\bibinfo{year}{2023}), \bibinfo{pages}{15790--15801}.
\newblock
\urldef\tempurl%
\url{https://doi.org/10.1109/TPAMI.2023.3305122}
\showDOI{\tempurl}


\bibitem[\protect\citeauthoryear{Wulff-Nilsen}{Wulff-Nilsen}{2017}]%
        {wulff2017fully}
\bibfield{author}{\bibinfo{person}{Christian Wulff-Nilsen}.}
  \bibinfo{year}{2017}\natexlab{}.
\newblock \showarticletitle{Fully-dynamic minimum spanning forest with improved
  worst-case update time}. In \bibinfo{booktitle}{\emph{Proceedings of the 49th
  Annual ACM SIGACT Symposium on Theory of Computing}}.
  \bibinfo{pages}{1130--1143}.
\newblock


\bibitem[\protect\citeauthoryear{Xing and Karp}{Xing and Karp}{2001}]%
        {DBLP:conf/ismb/XingK01}
\bibfield{author}{\bibinfo{person}{Eric~P. Xing} {and}
  \bibinfo{person}{Richard~M. Karp}.} \bibinfo{year}{2001}\natexlab{}.
\newblock \showarticletitle{{CLIFF:} clustering of high-dimensional microarray
  data via iterative feature filtering using normalized cuts}. In
  \bibinfo{booktitle}{\emph{Proceedings of the Ninth International Conference
  on Intelligent Systems for Molecular Biology, July 21-25, 2001, Copenhagen,
  Denmark}}. \bibinfo{pages}{306--315}.
\newblock


\bibitem[\protect\citeauthoryear{Yu and Shi}{Yu and Shi}{2003}]%
        {DBLP:conf/iccv/YuS03}
\bibfield{author}{\bibinfo{person}{Stella~X. Yu} {and} \bibinfo{person}{Jianbo
  Shi}.} \bibinfo{year}{2003}\natexlab{}.
\newblock \showarticletitle{Multiclass Spectral Clustering}. In
  \bibinfo{booktitle}{\emph{9th {IEEE} International Conference on Computer
  Vision {(ICCV} 2003), 14-17 October 2003, Nice, France}}.
  \bibinfo{publisher}{{IEEE} Computer Society}, \bibinfo{pages}{313--319}.
\newblock
\urldef\tempurl%
\url{https://doi.org/10.1109/ICCV.2003.1238361}
\showDOI{\tempurl}


\bibitem[\protect\citeauthoryear{Zhang, Xie, Feng, and Zhang}{Zhang
  et~al\mbox{.}}{2009}]%
        {DBLP:conf/airs/ZhangXFZ09}
\bibfield{author}{\bibinfo{person}{Jin Zhang}, \bibinfo{person}{Lei Xie},
  \bibinfo{person}{Wei Feng}, {and} \bibinfo{person}{Yanning Zhang}.}
  \bibinfo{year}{2009}\natexlab{}.
\newblock \showarticletitle{A Subword Normalized Cut Approach to Automatic
  Story Segmentation of Chinese Broadcast News}. In
  \bibinfo{booktitle}{\emph{Information Retrieval Technology, 5th Asia
  Information Retrieval Symposium, {AIRS} 2009, Sapporo, Japan, October 21-23,
  2009. Proceedings}} \emph{(\bibinfo{series}{Lecture Notes in Computer
  Science}, Vol.~\bibinfo{volume}{5839})},
  \bibfield{editor}{\bibinfo{person}{Gary~Geunbae Lee}, \bibinfo{person}{Dawei
  Song}, \bibinfo{person}{Chin{-}Yew Lin}, \bibinfo{person}{Akiko~N. Aizawa},
  \bibinfo{person}{Kazuko Kuriyama}, \bibinfo{person}{Masaharu Yoshioka}, {and}
  \bibinfo{person}{Tetsuya Sakai}} (Eds.). \bibinfo{publisher}{Springer},
  \bibinfo{pages}{136--148}.
\newblock
\urldef\tempurl%
\url{https://doi.org/10.1007/978-3-642-04769-5\_12}
\showDOI{\tempurl}


\bibitem[\protect\citeauthoryear{Zhao, Wang, and Feng}{Zhao
  et~al\mbox{.}}{2018}]%
        {zhao2018nearly}
\bibfield{author}{\bibinfo{person}{Zhiqiang Zhao}, \bibinfo{person}{Yongyu
  Wang}, {and} \bibinfo{person}{Zhuo Feng}.} \bibinfo{year}{2018}\natexlab{}.
\newblock \showarticletitle{Nearly-linear time spectral graph reduction for
  scalable graph partitioning and data visualization}.
\newblock \bibinfo{journal}{\emph{arXiv preprint arXiv:1812.08942}}
  (\bibinfo{year}{2018}).
\newblock


\end{thebibliography}

\appendix
\renewcommand{\thefigure}{\thesection.\arabic{figure}}
\renewcommand{\thetable}{\thesection.\arabic{table}}

\section{Theoretical Analysis of the Expander Decomposition}
\label{sec:proof}
The key ingredient for the analysis is to show that it is
very unlikely that the cut procedure does not return a cut within $T$ iterations
when started on a graph with conductance less than $\phi$. Consequently, if the
algorithm does not return a cut for $T$ iterations we can declare $G$ to be a
$\phi$-expander (or actually $G\{A\}$ to be a near $6\phi$-expander), with a
small probability of error.

For this we analyze the random walks in terms of flows. Each node $v$ injects
$d_v$ units of flow of a unique commodity. This flow is distributed according
to the random walk. 
Let $F_{ij}(t)$ denote the amount of
flow from node $j$ that has reached node $i$ after $t$ steps. We define
$P_{ij}(t) \defeq \frac{F_{ij}(t)}{d_i d_j}$ for all $t$. Whenever clear from context,
we may omit the explicit indication of the round $t$.

Let $A(t) \subseteq V$ be the set of nodes in the subgraph in round $t$. We define a
natural average vector
$\mu(t) \defeq \frac{1}{\vol{A(t)}}\sum_{i \in A(t)} d_i P_i(t)$ and track the
convergence of our random walk to this stationary distribution with a potential
function:
\[ \varphi(t) \defeq \sum_{i \in A(t)} d_i \cdot \norm{P_i(t) -
  \mu(t)}^2\enspace.
\]
Intuitively, a low potential indicates that the $P(t)$ vectors have mixed well
in the set $A(t)$. The following lemma shows that a low
potential implies that $A(t)$ is a near $6\phi$-expander in the graph $G$
($G\{A(t)\}$ may not be a proper expander because the random walk may have used edges outside of $G\{A(t)\}$).
\begin{lemma}\label{lemma:near-phi-expander}
    If $\varphi(t) \leq \frac{1}{4\vol{V}^2}$ in any step $t \leq T$, then 
    $A(t)$ is a near $6\phi$-expander in $G$.
\end{lemma}
\begin{proof}
We omit the time step $t$ to avoid notational clutter.
The random walk can be viewed as establishing a multicommodity
flow within the network. $d_iP_{ij}$ is the flow of commodity $j$ that reached
node $i$. The flow has congestion at most $t$---the number of steps of the
random walk. We take at most $T=1/(12\phi)$ steps.

Now consider a cut $S\subseteq A$, with $\nvol(S)\le\nvol(A)/2$. We have
to show that $|E(S,V\setminus S)|\ge \nvol(S)/(2T)$.
 For the potential to be less than ${1}/({4\vol{V}^2})$,
    we must have $P_{ij} \le \mu_j + \frac{1}{2 \vol{V}}$ for every node
    $i\in A$ and any $j$.
    The flow that originates in $S$ and stays in $S$ is
    \begin{equation*}
    \begin{split}%
    \sum_{j\in S}\sum_{i\in S}d_id_jP_{ij}
           &\le\sum_{j\in S}\sum_{i\in S}d_id_j\big(\mu_j+\tfrac{1}{2 \nvol(V)}\big)\\
           &\textstyle=\nvol(S)\sum_{j\in S}d_j\mu_j+\frac{\nvol(S)^2}{2\nvol(V)}\\
           &\textstyle\le \frac{\nvol(S)^2}{\nvol(A)}+\frac{\nvol(S)^2}{2\nvol(V)}\le \tfrac{3}{4}\nvol(S)\enspace.
    \end{split}
    \end{equation*}
    Here the third step follows because $d_j\mu_j\le d_j/\nvol(A)$ as $\mu_j$ is the
    fraction of the flow of commodity $j$ that stays in $A$ (normalized by
    $1/\nvol(A)$). This means that at least $1/4$ of the flow that starts in
    $S$ has to leave $S$. At the same time an equivalent amount of flow has to
    enter $S$, which means that the traffic across the edges in $E(S,V\setminus
    S)$ is at least $\nvol(S)/2$. As the congestion is only $T$, there must be at
    least $\frac{1}{2T}\nvol(S)\ge6\phi\nvol(S)$ edges across the cut.
\end{proof}

In the following we show that with constant probability during one iteration we
either return a balanced low conductance cut or the potential decreases
significantly. For this we first have to analyze by how much the potential
decreases during a random walk step.

\subsubsection*{Potential Decrease by Random Walk}

Fix a round $t$. Note that the random walk for round $t$ only takes the $t-1$ random walk steps for the graphs $G\{A(1)\}, \dots, G\{A(t-1)\}$, the step for $G\{A(t)\}$ follows in the next round.
In the following we develop an expression for how much the potential will decrease due to the random walk step in the current iteration $t$.
For this we need a technical claim, which is proven in
Section~\ref{section:deferred-proofs}.
\begin{claim}\label{claim:technical-averaging}
    Let $a_1,\dots,a_d,\mu$ be vectors of the same dimension. Then,
    \[  \textstyle
    d \big\|\tfrac{1}{d}\sum_{i} a_i  - \mu\big\|^2 - \sum_{i} \norm{a_i - \mu}^2
        = \frac{1}{d} \norm{\sum_{i} a_i}^2 - \sum_{i} \norm{a_i}^2
        \leq 0.
    \]
    Specifically, for $d=2$ the value is equal to $-\frac{1}{2}\norm{a_1 - a_2}^2$.
\end{claim}

\noindent
Let $\delta(t)\defeq(\varphi(t)-\varphi(t+1))/\varphi(t)$ be the relative factor by which the potential of round $t$ decreases after the random walk step in graph $G\{A(t)\}$.

\begin{lemma}\label{lem:potdecrease}
The relative potential decrease due to the random walk step in iteration $t$ is at least 
\begin{equation*}
\delta(t):=\frac{1}{2}\frac{\sum_{\{i,j\}\in G\{A(t)\}}\|P_i(t)-P_j(t)\|^2_2}{\sum_{i\in A(t)}d_i\|P_i(t)-\mu(t)\|_2^2}\enspace.
\end{equation*}
\end{lemma}
\begin{proof}
To simplify the analysis of a random walk step we view each vertex $v$ of
$G\{A(t)\}$ as consisting of $d_v$ sub-vertices as follows. For each edge
$\{x,y\}$ we introduce two sub-vertices -- one at $x$ and one at $y$, that
are connected by an edge. In this way, every super-vertex $i$ from $A(t)$
receives $d_i$ sub-vertices and every sub-vertex has a unique neighbour.

During a random walk step, each sub-vertex $x$ first takes a
copy of the flow vector of its super-vertex. Then the vectors are averaged
along every edge, and finally a super-vertex computes the average of the
vector of its sub-vertices.

\def\sup{\operatorname{sup}}
\def\sub{\operatorname{sub}}
We can express the potential by either summing over super-vertices or
sub-vertices:
\begin{equation*}
\varphi(t)=\sum_{i\in A(t)}d_i\|P_i-\mu\|_2^2 = \sum_x\|P_{\sup(x)}-\mu\|_2^2\enspace,
\end{equation*}
where $\operatorname{sup}(x)\in A(t)$ is the super-vertex for sub-vertex $x$.
Averaging the vectors between two sub-vertices $x$ and $y$ changes the potential
by
\begin{equation}\label{eqn:firstavg}
\begin{split}
2&\|\tfrac{1}{2}(P_{\sup(x)}+P_{\sup(y)})-\mu\|_2^2\\
&-\|P_{\sup(x)}-\mu\|_2^2-\|P_{\sup(y)}-\mu\|_2^2\enspace,
\end{split}
\end{equation}
which is equal to $-\frac{1}{2}\|P_{\sup(x)}-P_{\sup(y)}\|_2^2$ by
Claim~\ref{claim:technical-averaging}. Let $\psi$ denote the potential
after averaging along edges but before averaging
among sub-vertices. Similarly, let for a sub-node $x$, $Q_x$ denote the 
vector of $x$ at this point.
Summing Equation~\ref{eqn:firstavg} over all edges gives 
\begin{equation*}\textstyle
\psi-\varphi(t)\le\sum_{\{x,y\}\in E(A(t))} -\tfrac{1}{2}\|P_{\sup(x)}-P_{\sup(y)}\|_2^2
\end{equation*}
and
\begin{equation*}
\begin{split}
    \psi&=\sum_x\|Q_x-\mu\|_2^2
    =\sum_{i\in A(t)}\sum_{x\in\sub(i)}\|Q_x-\mu\|_2^2 \\
    &\ge \textstyle\sum_{i\in A(t)}d_i\|\tfrac{1}{d_i}\!\!\textstyle\sum_{x\in\sub(i)}Q_x-\mu\|_2^2\enspace,
\end{split}
\end{equation*}
where the inequality also follows from Claim~\ref{claim:technical-averaging}.
This means averaging among edges decreased the potential by
$\frac{1}{2}\sum_{\{i,j\}}\|P_i-P_j\|_2^2$ and averaging among sub-nodes
does not increase it again. Hence, the lemma follows.
\end{proof}

\subsubsection*{Projected Potential Decrease}
The algorithm does not maintain the $P_i$ vectors or the potential explicitly,
as that would be computationally infeasible
Rather it operates on their
projections onto some random vector $r$. The entries in the vector $u$ of the algorithm are
actually $u_i = (P_i - \mu)^\top r$.
This follows since performing the random walk and then projecting the resulting $P_i$ vectors onto $r$ is equivalent to first projecting onto $r$ and then running the random walk on the vector of projections. Ultimately, subtracting the weighted average of the $u$ values corresponds to subtracting $\mu$ from the  $P_i$ vectors before the projection.

Consider the quotient
\begin{equation*}
R(u)\defeq \frac{\sum_{\{i, j\} \in E(A(t))} (u_i - u_j)^2}{\sum_{i \in A(t)} d_i u_i^2}\enspace,
\end{equation*}
and compare it to the bound $\delta(t)$ for the relative potential decrease in
Lemma~\ref{lem:potdecrease}.
Up to a factor of 1/2, $R(u)$ is obtained by individually performing a random projection on
each vector in the expression for $\delta(t)$. Using properties of random
projections, we will show that $R(u)$ is a good indicator for the potential
decrease $\delta(t)$.%

\begin{lemma}[Projection Lemma]\label{lem:projection}
    Let $v_1,\dots,v_n$ be a collection of $d$-dimensional vectors. Let $r$
    denote a random $d$-dimensional vector
    with each coordinate sampled independently from a Gaussian distribution
    $\mathcal{N}(0,1/d)$. Then
    \begin{enumerate}
        \item $\Pr[\exists i,j: (v_i-v_j)^\top r \ge 11\log n\cdot \|v_i-v_j\|^2/d]\le 1/n$ \label{ProC}
        \item $\Pr[\sum_i(v_i ^\top r)^2 \geq \frac{1}{20\log
          n}\sum_i\|v_i\|^2/d] \ge 1/2$ for $n\ge 8$. \label{ProD}
    \end{enumerate}
\end{lemma}
We call an iteration \emph{good} if the vector of projections $u$ retains sufficient information of the higher dimensional $P_i$ vectors. Formally, we require the following.

\begin{definition}
An iteration $t$ of the algorithm is \emph{good} if
\begin{itemize}
    \item $\sum_{i \in A(t)} d_i u_i^2 \geq \tfrac{1}{20 \, n\log{n}} \cdot \sum_id_i\|P_i-\mu\|^2 \quad$ and
    \item $(u_i - u_j)^2 \leq \tfrac{11\log{n}}{n} \norm{P_i - P_j}^2 \quad \forall\{i,j \} \in E(A(t))$.
\end{itemize}
\end{definition}

From Property~\ref{ProC} and Property~\ref{ProD} of the Projection Lemma we can directly infer the following.

\begin{claim}\label{claim:round-good}
    An iteration is good with probability at least $1/4$ for $n\ge 4$.
\end{claim}

The next lemma shows that 
in a good round a large  $R(u)$ value implies a large $\delta(t)$ value, i.e.,
a large potential decrease.
\begin{lemma}\label{lemma:Ru-less-than-delta}
In a good round $t$ we have
     $R(u) \leq 440\log^2(n) \cdot \delta(t) $.
\end{lemma}
\begin{proof}
\begin{equation*}
\begin{split}
\delta(t)
&=\frac{1}{2}\frac{\sum_{\{i,j\}\in G\{A(t)\}}\|P_i-P_j\|^2_2}{\sum_{i\in A(t)}d_i\|P_i-\mu\|_2^2} \\
&\ge \frac{1}{2}\frac{\frac{n}{11\log n}\sum_{\{i,j\}\in G\{A(t)\}}(u_i-u_j)^2}{20\,n\log n\sum_{i\in A(t)}d_iu_i^2}
\end{split}
\end{equation*}
by the definition of $\delta(t)$ and the definition of a good round.
\end{proof}

\noindent
{\bfseries For the further analysis we will always assume that we are in a good round.}
\subsubsection*{Finding a Cut}
With the above lemma we have established, that a large $R(u)$ value ensures the
next step of the random walk reduces the potential a lot. Now we show that a
small $R(u)$ value ensures that the algorithm finds a low conductance cut.

Indeed the quotient $R(u)$ (the Rayleigh quotient of $u$ using the graph
Laplacian), and its analysis lies at the heart of the proof  
for the celebrated Cheeger inequality. It states that a small $R(u)$ value implies the existence
of a low conductance sweep cut.

\begin{lemma}[Cheeger]\label{lemma:cheeger}
    For any $x \in \mathbb{R}^n$, with $x \perp \vec{d}$ there is a value $c$, s.t. there is a non-trivial set $S_c = \{ i \in V: x_i \leq c \}$ with conductance $\Phi_G(S_c) \leq \sqrt{2R(x)}$.
\end{lemma}

Note that $u\perp\vec{d}$ holds by design of the algorithm due to Line~10 of the algorithm.
With this Cheeger Lemma we thus get the guarantee that there has to be a sweep cut through $u$ whose conductance $\Phi$ satisfies $\Phi \leq \sqrt{ 2R(u)}$.
Together with Lemma~\ref{lemma:Ru-less-than-delta} this now implies the following:
If we will only make little progress in the potential with the next random walk step, i.e., $\delta(t) \leq \alpha$ for some small $\alpha$, then we find a sweep cut through the vector of projections $u$ with conductance $\Phi \leq \mathcal{O}(\sqrt{\alpha} \log{n})$. By contraposition, if all sweep cuts have conductance at least $\Phi$, then $\delta(t) \geq \Omega(\Phi^2/\log^2{n})$.

This motivates our approach of analyzing the sweep cuts on $u$ and differentiating three cases: 1) No low conductance sweep cut exists, 2) A balanced low conductance cut exists, or 3) An unbalanced low conductance cut exists.

If we find a balanced sweep cut, we return immediately as this ensures the recursion depth of the surrounding expander decomposition remains small. If no sweep cut has conductance below $\gamma$, the above reasoning ensures we make sufficient progress with the next random walk step.
Next we prove that even if we find an unbalanced cut, the potential still decreases sufficiently. 

\subsubsection*{Handling Unbalanced Cuts}
Let $(S,B)$ be an unbalanced cut, where $\vol{S} \leq 2\log^2(\vol{A})/\vol{A}$ and there is no $\gamma$-low conductance sweep cut through $B$ in $u$.
We show that if $R(u)$ is small and we thus cannot ensure sufficient progress from the random walk, then the set $S$ must contain a large fraction of the potential. Hence we instead make progress by removing $S$ when we set $A(t+1) = B$.

Consider the following lemma, which gives an upper bound on how much the $R(u)$ value can at most increase if we remove a cut and recenter the remaining projected values.
Recentering the vector after removing the cut allows us to use Lemma~\ref{lemma:cheeger} to argue about the existence of low conductance cuts in the remaining subgraph.
Let $z_B \in \mathbb{R}^{\vert B \vert}$ be the shortened and recentered vector of the projections $u$, i.e., $z_i \defeq u_i - \overline{u}_B$ for all $i \in B$ \footnote{
For any set $X \subseteq V$ and vector $x \in \mathbb{R}^{\vert X \vert}$ we denote by $\overline{x}_X \defeq \frac{1}{\vol{X}} \sum_{i \in X} d_i x_i$ the degree- weighted average of the $x$ values in $X$.}. By design, it then holds $z_B \perp \vec{d}$.

\begin{lemma}\label{lemma:unbalanced-cut-R-values}
    For any cut $(S,B)$ in $A$ and $\lambda < \vol{B}/\vol{A}$ such that $\sum_{i \in S} d_i u_i^2 \leq \lambda \cdot \sum_{i \in A} d_i u_i^2$,
    \[  \textstyle
        R(z_B) \leq \frac{\vol{B}}{\vol{B} - \lambda\vol{A}} \cdot R(u).
    \]
\end{lemma}
The proof is deferred to Section~\ref{section:deferred-proofs}.
We can put this in the context of our unbalanced cut $(S,B)$ with the following lemma. For this, let $\varphi_S$ and $\varphi_B$ denote the potential generated by summing only over the nodes in $S$ and $B$ respectively.

\begin{lemma}\label{lemma:unbalanced-cut-progress}
    If $R(u) < \gamma^2/4$ and we find an unbalanced cut $(S,B)$ with $\vol{S} \leq 2m/\log^2{m}$,
    then $\varphi_B(t) \leq \big(1 - 1/(880 \log^2{n}) \big) \cdot \varphi(t)$.
\end{lemma}
\begin{proof}
Since $(S,B)$ partitions $A$ we get $\varphi_S(t) + \varphi_B(t) = \varphi(t)$.
Assume for contradiction that $\varphi_S(t) < 1/(880 \log^2{n}) \cdot \varphi(t)$, i.e., with $S$ we only remove very little potential. It follows
\begin{equation}\label{eq:bound-on-S}
    \textstyle \sum_{i \in S} d_i u_i^2 
    \leq \textstyle \frac{11\log{n}}{n} \varphi_S(t)
    < \textstyle \frac{1}{80n\log{n}} \varphi(t)
    \leq \textstyle \frac{1}{4}  \sum_{i \in A} d_i u_i^2
\end{equation}
since round $t$ is good.
Note that $\frac{\vol{B}}{\vol{A}} \geq 1- \frac{2\log^2{(\vol{A})}}{\vol{A}} \geq \frac{3}{4}$ and therefore
\begin{equation}\label{eq:fractions-simplify}
\begin{aligned}
    \textstyle\frac{\vol{B}}{\vol{B} - 1/4\cdot \vol{A}} \leq \frac{\vol{A}}{(3/4 - 1/4) \vol{A}} \leq 2.
\end{aligned}
\end{equation}
Summing over $B$ gives $\sum_i d_i z_i = \sum_i d_i u_i - \vol{B}\overline{u}_B = 0$. This allows us to use Lemma~\ref{lemma:cheeger} in the subgraph $G\{B\}$.
By design of the algorithm, since we did not find a balanced low conductance cut, $(S,B)$ is a two-ended sweep cut where $S$ has maximal volume among low conductance cuts.
Hence there is no low conductance sweep cut through $u$ that crosses the set $B$.
We get
\begin{align*}
    \gamma^2 & \leq 2R(z_B) && \text{by Lemma~\ref{lemma:cheeger}} \\
    & \leq 2\textstyle\frac{\vol{B}}{\vol{B} - 1/4\cdot\vol{A}} \cdot R(u) && \text{by Lemma~\ref{lemma:unbalanced-cut-R-values} using (\ref{eq:bound-on-S})} \\
    & \leq 4 R(u) && \text{by (\ref{eq:fractions-simplify}).}
\end{align*}
This is a contradiction to our assumption that $R(u) < \gamma^2/4$.

\end{proof}

The above results show that a good round ensures sufficient progress in decreasing the potential or produces a balanced cut.

\begin{claim}\label{claim:good-round-reduction}
    In a good round we either find a balanced cut or the potential reduces by a factor at least $1-\gamma^2/(1760\log^2{n})$.
\end{claim}
\begin{proof}
    Let $\Phi$ be the minimum value of any sweep cut on $u$.
    
    If $\Phi \geq \gamma$, then Lemmas~\ref{lemma:Ru-less-than-delta} and~\ref{lemma:cheeger} give $\delta \geq \gamma^2/(880\log^2{n})$ since
    \[
        \gamma^2 \leq \Phi^2 \leq 2R(u) \leq 880 \log^2(n) \cdot \delta \enspace.
    \]
    
    If $\Phi < \gamma$, we will find some low conductance cut $S$. If $S$ is balanced, we can return it, otherwise we move $S$ from $A$ to $R$. Then either
    \begin{itemize}
        \item $R(u) \geq \gamma^2/4$ and Lemmas~\ref{lemma:Ru-less-than-delta} and~\ref{lemma:cheeger} give $\delta \geq \gamma^2/(1760\log^2{n})$, or
        \item $R(u) < \gamma^2/4$ and Lemma~\ref{lemma:unbalanced-cut-progress} gives a factor of $1 - 1/(880\log^2{n})$.
    \end{itemize}
    In any case, the decrease is at least $1-\gamma^2/(1760\log^2{n})$.
\end{proof}

\subsubsection*{Iterations}
Computing one step of the random walk $\vec{u}$ can be done by averaging the values at the endpoints of each edge. In iteration $t$, we can thus compute the values $\vec{u}$ in $\mathcal{O}(mt)$ time. Sorting the $\vec{u}$ values takes $\mathcal{O}(n\log{n})$ time. 
The following lemma shows that there are at most $O(1/\phi)$ iterations.
With this, the Cut Procedure takes $O((m \log{n}) / \phi^2)$ time in total.

It only remains to prove the upper bound on the number of iterations. The threshold for the potential is chosen to guarantee the existence of a near $6\phi$-expander in Lemma~\ref{lemma:near-phi-expander}.

\begin{lemma}\label{lemma:numer-of-iterations}
    After $T = 1 / 12\phi$, iterations of the cut procedure, with high probability $\varphi(T) \leq \frac{1}{4\vol{V}^2}$.
\end{lemma}
\begin{proof}
    Let $\nvol = \vol{V}$. The initial potential is
    \begin{align*}
        \varphi(0) & = \textstyle \sum_i d_i \cdot \big(  \textstyle (n-1)\cdot \frac{1}{\nvol^2} + (1 - \frac{1}{\nvol})^2 \big) \\
        & = \nvol \cdot \big(\textstyle \frac{n}{\nvol^2} - \frac{2}{\nvol} + 1 \big) \leq \nvol.
    \end{align*}
    If we reach round $T$, then none of the previous iterations produced a balanced cut. By Claim~\ref{claim:good-round-reduction}, a good round decreases the potential by a factor of $1-\gamma^2/(1760\log^2{n})$. We first show that reducing the initial potential by this factor $\frac{1}{96\phi}$ times brings it below $\frac{1}{4\nvol^2}$, then we show that with high probability there are sufficient good rounds.
    After $k$ good rounds, the potential has decreased by a factor of
    \begin{align*}
        \textstyle
        \Big( 1-\frac{\gamma^2}{1760\log^2{n}} \Big)^k 
        < \textstyle \operatorname{exp}\big( -\frac{\gamma^2}{1760\log{n}}\big)^\frac{k}{\log{n}}
        & =  \textstyle \operatorname{exp}\big( -\frac{\gamma^2}{1760\log^2{n}} \cdot k \big)
    \end{align*}
    where the first inequality follows from $(1 + x/a)^a < e^x, \forall x\in \mathbb{R}, a > 0$ with $a=\log{n}$.
    Combining this with the above bound on $\varphi(0)$, we can take logarithms on both sides to see that $\varphi(T) \leq 1/(4\nvol^2)$ if $k \cdot \gamma^2/(1760\log^2{n}) \geq  \ln(4\nvol^3)$
    and thus we need at least%
    \[  \textstyle
        1760\frac{1}{\gamma^2}\ln(4\nvol^3)\log^2{n} \leq 1220 \frac{1}{\gamma^2}\log(4\nvol^3)\log^2(n)
    \]
    good rounds.
    Setting 
    \[
        \gamma = 343 \sqrt{\phi \log(4\nvol^3)\log^2(n) }
    \]
    gives that after $\frac{1}{96\phi}$ good rounds, the potential is below the threshold.

    By Claim~\ref{claim:round-good}, a round is good with probability $1/4$. Let $X_i$ be the indicator random variable representing whether round $i$ was good and $K = \sum_i^T X_i$ their sum. Clearly, $K$ is the number of good rounds and its expectation is $\EX[K] = T/4$.
    Since $\phi \leq 1/\log^2{n}$ we get $T \geq \log^2(n)/12$ and thus $\EX[K] \geq \log^2(n)/48$.
    Using a Chernoff bound then gives
    \begin{align*}
        \Pr \big[K \leq \textstyle\frac{1}{2} \EX[K] \big] \textstyle \leq \exp(-\frac{1}{8} \EX[K]) 
        \textstyle \leq n^{-\frac{\log{n}}{48}} \enspace .
    \end{align*}
    It follows that with high probability the number of good rounds is
    \[  \textstyle
        K \geq \textstyle\frac{1}{2} \EX[K] = \frac{T}{8} \geq \frac{1}{96\phi} \enspace.
    \]
\end{proof}

Altogether, the above arguments imply the correctness of our Theorem~\ref{theorem:cut-step}. Analogous to the reasoning by Saranurak and Wang~\cite{saranurakwang19}, we can conclude that our modified guarantees yield Theorem~\ref{theorem:expander-decomposition} when using our cut step procedure in their expander decomposition framework.

\subsection{Deferred Proofs}\label{section:deferred-proofs}
\begin{lemma}[Projection Lemma]
    Let $v_1,\dots,v_n$ be a collection of $d$-dimensional vectors. Let $r$
    denote a random $d$-dimensional vector
    with each coordinate sampled independently from a Gaussian distribution
    $\mathcal{N}(0,1/d)$. Then
    \begin{enumerate}
        \item $\EX[(v_i ^\top r)^2] = \|v_i\|^2/d$ for all $i$\label{ProA}
        \item $\Pr[(v_i ^\top r)^2 \ge \alpha \cdot \|v_i\|^2/d] \le e^{-\alpha/5}$ for
        $\alpha \ge 1$\label{ProB}
        \item $\Pr[\exists i,j: (v_i-v_j)^\top r \ge 11\log n\cdot \|v_i-v_j\|^2/d]\le 1/n$\label{ProC2}
        \item $\Pr[\sum_i(v_i ^\top r)^2 \geq \frac{1}{20\log
          n}\sum_i\|v_i\|^2/d] \ge 1/2$ for $n\ge 8$.\label{ProD2}
    \end{enumerate}
\end{lemma}

\begin{proof}[Proof of Lemma~\ref{lem:projection}]\hfill\\
\noindent\ref{ProA}.~~
Due to the rotation symmetry of the Gaussian distribution we can assume
wlog.\ that $v$ is an arbitrary vector of length $\ell=\|v\|$. In
particular we can assume that $v_1=\sqrt{\ell}$ and all other entries are $0$.
Then $\EX[(v^\top r)^2]=\EX[\ell\cdot\mathcal{N}(0,1/d)^2]=\ell/d\cdot\EX[\mathcal{N}(0,1)^2]=\ell/d$.

\ref{ProB}.~~
Using Lemma~1 from Laurent and Massart~\cite{LM2000} one gets that for a
$\chi$-squared distributed variable $X\sim\mathcal{N}(0,1)^2$ fulfills
$\Pr[X\ge 1+2\sqrt{x}+2x]\le e^{-x}$, which gives 
$\Pr[X\ge 5x]\le e^{-x}$ for $x\ge 1$. Since the distribution of $(v^\top r)^2$
is $\frac{\ell}{d}\cdot\mathcal{N}(0,1)^2$ the statement follows.

\ref{ProC2}.~~
Using Property~\ref{ProB} with $\alpha=15\ln n$ gives that the probability
that a vector is overstretched by a factor more than $\alpha$ is at most 
$n^{-3}$. Applying a union bound over the at most $n^2$ vectors gives the statement. Finally, note that $15\ln x < 11 \log_2 x$ for all $x > 1$.

\ref{ProD2}.~~
Let $x_j$ denote the vector where the $i$-th entry is $v_{ij}$.
Then
$Z:=\sum_i(v_i ^\top r)^2
=\sum_i\sum_j(v_{ij}r_j)^2
=\sum_j\sum_i(x_{ji}r_j)^2=\sum_j\|x_j\|_2^2\cdot (r_j)^2$.
Let $\ell_{\max}$ denote the largest length of an $x_j$-vector.
We classifiy the $x_j$-vectors whose length fall in the range
$(\ell_{\max}/n,\ell_{\max}]$ into $\log_2n$ classes so that the length
of vectors in the same class differs by at most a factor of $2$. Formally,
the $i$-th class for $i\in 1,\dots,\log_2n$ contains vectors with length
in the range $(\ell_{\max}/2^i,\ell_{\max}/2^{i-1}]$. Scale the length for
each classified vector $x_j$ down to the lowest length in its class and scale
the unclassified vectors down to length $0$. Let $\tilde{x}_j$ denote scaled
length of a vector $x_j$. We have
\begin{equation*}
\begin{split}
\textstyle{\sum_j\|\tilde{x}\|_j^2}
  &\ge \textstyle{\sum_{\text{$j$ classified}}\|x_j\|^2/4}
  \ge \textstyle{\sum_{j}\|x_j\|^2/4-n(\ell_{max}/n)^2}\\
  &\ge \textstyle{(\tfrac{1}{4}-\tfrac{1}{n})\sum_{j}\|x_j\|^2}
  \ge \textstyle{\tfrac{1}{8}\sum_{j}\|x_j\|^2}\enspace.
\end{split}
\end{equation*}\def\R{\mathcal{R}}%
for $n\ge 4$. The second inequality holds because an unclassified vector has at
most length $\ell_{\max}/n$ and the third inequality because there is one
vector $x_j$ with length $\ell_{\max}$. Now, we choose the class where the
total scaled length of vectors is largest. Let $\R$ denote the index
set of vectors in this class, let $\ell_{\R}$ denote their scaled
length and let $n_{\R}$ denote their number.

Then $n_\R\ell_\R^2\ge\tfrac{1}{8\log_2n}\sum_{j}\|x_j\|^2$. Now, we
consider the sum of squared projections of the scaled vectors in $\R$.
This is distributed according to a $\chi^2$-distribution with $n_\R$
degrees of freedom scaled by $\ell_\R^2/d$. The median for this distribution
is $n_\R(1-{2}/(9n_R))\ge 0.4n_R$. Hence, with probability $1/2$
\begin{equation*}
Z\ge Z_\R\ge 0.4n_\R\ell_\R^2/d\ge \tfrac{1}{20\log_2n}\sum_{j}\|x_j\|^2\enspace.
\end{equation*}

\end{proof}

\begin{proof}[Proof of Lemma~\ref{lemma:unbalanced-cut-R-values}]
    For the numerator of $R(z_B)$, we get
    \[ \textstyle
        \begin{split}\textstyle
            \sum_{\{ i, j\} \in E(B)} (z_i - z_j)^2 &= \textstyle \sum_{\{ i, j\} \in E(B)} (u_i - u_j)^2 \\ & \textstyle\leq \sum_{\{ i, j\} \in E(A)} (u_i - u_j)^2.
        \end{split}
    \]
    For the denominator, recall that $S$ and $B$ form a partition of $A$ and further $u \perp \vec{d}$.
    First, we get
    $0 = \sum_{i \in A} d_i u_i = \vol{S}\overline{u}_S + \vol{B}\overline{u}_B$, which can be rearranged to give
    \begin{equation}\label{eq:avg-B-squared}
        \textstyle
        \overline{u}_B^2 = \frac{\vol{S}^2}{\vol{B}^2}\cdot\overline{u}_S^2 \enspace.
    \end{equation}    
    Next, we observe that 
        \begin{equation}\label{eq:sum-S-avg}
        \textstyle
        \sum_{i \in S} d_i u_i^2 = \sum_{i \in S} d_i (u_i - \overline{u}_S)^2 + \vol{S}\overline{u}_S^2 \geq \vol{S}\overline{u}_S^2 \enspace.
    \end{equation}
    Putting everything together, we can bound the denominator of $R(z_B)$ with
    \begin{align*}
        \textstyle
        \sum_{i \in B} d_i z_i^2
        & = \textstyle \sum_{i \in B} d_i u_i^2 - \vol{B}\cdot\overline{u}_B^2 \\
        & \geq \textstyle \sum_{i \in B} d_i u_i^2 - \frac{\vol{S}}{\vol{B}} \sum_{i \in S} d_i u_i^2 && \text{by (\ref{eq:avg-B-squared}), (\ref{eq:sum-S-avg})} \\
        & \geq \textstyle (1-\lambda)\sum_{i \in A} d_i u_i^2 - \lambda \frac{\vol{S}}{\vol{B}}  \sum_{i \in A} d_i u_i^2  && \text{assumption} \\
        & = \textstyle \frac{\vol{B} - \lambda\vol{A}}{\vol{B}} \cdot \sum_{i \in A} d_i u_i^2 && \text{$A = B \cup S$}
    \end{align*}
    and therefore $R(z_B) \leq \frac{\vol{B}}{\vol{B} - \lambda\vol{A}} \cdot R(u)$.
\end{proof}

\begin{proof}[Proof of Claim~\ref{claim:technical-averaging}]
    \begin{align*}
        &\textstyle d \norm{\frac{\sum_{i} a_i}{d}  - \mu}^2 - \sum_{i} \norm{a_i - \mu}^2 \\
        & = d\Big( \textstyle \norm{\frac{1}{d}\sum_i a_i}^2 - \frac{2}{d}\sum_i a_i^\top\mu + \norm{\mu}^2  \Big)  - \sum_{i} \norm{a_i - \mu}^2 \\
        & \textstyle = \frac{1}{d} \norm{\sum_i a_i}^2 - \cancel{2\sum_i a_i^\top\mu} + \cancel{d\norm{\mu}^2}
        - \sum_i \big( \norm{a_i}^2 - \cancel{2a_i^\top\mu} + \cancel{\norm{\mu}^2} \big) \\
        & = \textstyle \frac{1}{d} \underbrace{\textstyle \norm{\sum_i a_i}^2}_{\leq \sum_{i} \norm{a_i}^2} - \sum_{i} \norm{a_i}^2 \leq 0.
    \end{align*}
    For $d=2$ we get
    \begin{align*}
        & \textstyle \frac{1}{2} 
        \norm{a_1 + a_2}^2 - \norm{a_1}^2 - \norm{a_2}^2 \\
        & = \textstyle\frac{1}{2} \norm{a_1}^2 + a^\top b +  \frac{1}{2} \norm{a_2}^2  - \norm{a_1}^2 - \norm{a_2}^2 \\
        & = \textstyle -\frac{1}{2}\norm{a_1 - a_2}^2.
    \end{align*}
\end{proof}

 \section{Additional Related Work} %
\paragraph{Multilevel Graph Partitioning}\label{sec:mgp}
This framework has been employed by many successful graph partitioning tools, e.g., \Metis{}~\cite{metis}, \Graclus{}~\cite{graclus}, or \Kahip{}~\cite{kahip}.
The general idea of the paradigm is to compute a small summary (``coarsening'') of the graph, on which it is easier to solve the problem we are interested in, and then mapping this coarse solution onto the original graph.
In more technical terms, a solver following this paradigm has three phases:
\emph{(i) Coarsening}: Compute a series of successively smaller graphs $G_0, G_1, \dots, G_\ell$ such that $|V_0| > |V_1| > \dots > |V_\ell|$ with the aim to create a coarse summary of the original graph.
\emph{(ii) Solving}: Obtain a solution to the initial problem on the graph $G_\ell$. Although $G_\ell$ is small, heuristics are usually employed here.
\emph{(iii) Refinement}: Successively map the solution from graph $G_{i}$ to graph $G_{i-1}$, while applying heuristics that increase the solution quality.

 \section{Instance List} %
See \autoref{tab:instances} for a full list of instances and information on their degree distribution.

\begin{table*}
\caption{Statistics on the number of vertices, edges and degree distribution of the graph instances in our benchmark dataset.
$\MaxDegree$ denotes the maximum, $d_\textsf{mean}$ the average vertex degree.
The last four columns are percentiles of the degree distribution.}\label{tab:instances}
\setlength{\tabcolsep}{3pt}
\begin{tabular}{lllrrrrrrrr}
\toprule
\# & Name & Type & $|V|$ & $|E|$ & $\MaxDegree$ & $d_\textsf{mean}$ & 25th & 50th & 75th & 90th \\
\midrule
BP1 & imdb & Bipartite & \num{1403278} & \num{4303383} & \num{1652} & \num{6.13} & \num{1} & \num{2} & \num{6} & \num{15} \\
CF1 & ramage02 & Computational Fluids & \num{16830} & \num{1424761} & \num{269} & \num{169.31} & \num{131} & \num{170} & \num{170} & \num{269} \\
CL1 & uk & Clustering & \num{4824} & \num{6837} & \num{3} & \num{2.83} & \num{3} & \num{3} & \num{3} & \num{3} \\
CL2 & smallworld & Clustering & \num{100000} & \num{499998} & \num{17} & \num{10.00} & \num{9} & \num{10} & \num{11} & \num{12} \\
CN1 & citationCiteseer & Citation Network & \num{268495} & \num{1156647} & \num{1318} & \num{8.62} & \num{2} & \num{5} & \num{10} & \num{18} \\
CN2 & ca-hollywood-2009 & Citation Network & \num{1069126} & \num{56306653} & \num{11467} & \num{105.33} & \num{13} & \num{31} & \num{75} & \num{212} \\
CN3 & coAuthorsDBLP & Citation Network & \num{299067} & \num{977676} & \num{336} & \num{6.54} & \num{2} & \num{4} & \num{7} & \num{14} \\
CN4 & ca-MathSciNet & Citation Network & \num{332689} & \num{820644} & \num{496} & \num{4.93} & \num{1} & \num{3} & \num{5} & \num{11} \\
CN5 & ca-coauthors-dblp & Citation Network & \num{540486} & \num{15245729} & \num{3299} & \num{56.41} & \num{13} & \num{34} & \num{74} & \num{135} \\
CN6 & ca-dblp-2012 & Citation Network & \num{317080} & \num{1049866} & \num{343} & \num{6.62} & \num{2} & \num{4} & \num{7} & \num{14} \\
CN7 & ca-citeseer & Citation Network & \num{227320} & \num{814134} & \num{1372} & \num{7.16} & \num{2} & \num{4} & \num{8} & \num{15} \\
CN8 & coPapersCiteseer & Citation Network & \num{434102} & \num{16036720} & \num{1188} & \num{73.88} & \num{15} & \num{39} & \num{92} & \num{177} \\
CN9 & ca-dblp-2010 & Citation Network & \num{226413} & \num{716460} & \num{238} & \num{6.33} & \num{2} & \num{4} & \num{7} & \num{13} \\
CS1 & add32 & Circuit Simulation & \num{4960} & \num{9462} & \num{31} & \num{3.82} & \num{2} & \num{3} & \num{4} & \num{9} \\
CS2 & rajat10 & Circuit Simulation & \num{30202} & \num{50101} & \num{101} & \num{3.32} & \num{2} & \num{4} & \num{4} & \num{4} \\
CS3 & memplus & Circuit Simulation & \num{17758} & \num{54196} & \num{573} & \num{6.10} & \num{2} & \num{3} & \num{4} & \num{11} \\
DM1 & pcrystk02 & Duplicate Materials & \num{13965} & \num{477309} & \num{80} & \num{68.36} & \num{53} & \num{80} & \num{80} & \num{80} \\
EM1 & email-enron-large & Email Network & \num{33696} & \num{180811} & \num{1383} & \num{10.73} & \num{1} & \num{3} & \num{7} & \num{19} \\
EM2 & email-EU & Email Network & \num{32430} & \num{54397} & \num{623} & \num{3.35} & \num{1} & \num{1} & \num{1} & \num{3} \\
FE1 & fe\_tooth & Finite Elements & \num{78136} & \num{452591} & \num{39} & \num{11.58} & \num{6} & \num{12} & \num{15} & \num{18} \\
FE2 & fe\_rotor & Finite Elements & \num{99617} & \num{662431} & \num{125} & \num{13.30} & \num{11} & \num{13} & \num{14} & \num{17} \\
IF1 & inf-openflights & Infrastructure Network & \num{2939} & \num{15677} & \num{242} & \num{10.67} & \num{2} & \num{3} & \num{8} & \num{28} \\
IF2 & inf-power & Infrastructure Network & \num{4941} & \num{6594} & \num{19} & \num{2.67} & \num{2} & \num{2} & \num{3} & \num{5} \\
NS1 & wing\_nodal & Numerical Simulation & \num{10937} & \num{75488} & \num{28} & \num{13.80} & \num{12} & \num{14} & \num{16} & \num{17} \\
NS2 & auto & Numerical Simulation & \num{448695} & \num{3314611} & \num{37} & \num{14.77} & \num{13} & \num{15} & \num{16} & \num{18} \\
OP1 & gupta2 & Optimization & \num{62064} & \num{2093111} & \num{8412} & \num{67.45} & \num{25} & \num{38} & \num{47} & \num{49} \\
OP2 & finance256 & Optimization & \num{37376} & \num{130560} & \num{54} & \num{6.99} & \num{4} & \num{6} & \num{8} & \num{12} \\
RD1 & appu & Random Graph & \num{14000} & \num{919552} & \num{293} & \num{131.36} & \num{109} & \num{131} & \num{154} & \num{176} \\
RN1 & inf-roadNet-CA & Road Network & \num{1957027} & \num{2760388} & \num{12} & \num{2.82} & \num{2} & \num{3} & \num{3} & \num{4} \\
RN2 & inf-roadNet-PA & Road Network & \num{1087562} & \num{1541514} & \num{9} & \num{2.83} & \num{2} & \num{3} & \num{4} & \num{4} \\
RN3 & inf-italy-osm & Road Network & \num{6686493} & \num{7013978} & \num{9} & \num{2.10} & \num{2} & \num{2} & \num{2} & \num{3} \\
RN4 & luxembourg\_osm & Road Network & \num{114599} & \num{119666} & \num{6} & \num{2.09} & \num{2} & \num{2} & \num{2} & \num{3} \\
SN1 & soc-youtube-snap & Social Network & \num{1134890} & \num{2987624} & \num{28754} & \num{5.27} & \num{1} & \num{1} & \num{3} & \num{8} \\
SN2 & soc-flickr & Social Network & \num{513969} & \num{3190452} & \num{4369} & \num{12.41} & \num{1} & \num{1} & \num{5} & \num{17} \\
SN3 & soc-lastfm & Social Network & \num{1191805} & \num{4519330} & \num{5150} & \num{7.58} & \num{1} & \num{2} & \num{4} & \num{11} \\
SN4 & soc-twitter-follows & Social Network & \num{404719} & \num{713319} & \num{626} & \num{3.53} & \num{1} & \num{1} & \num{1} & \num{3} \\
SN5 & soc-pokec & Social Network & \num{1632803} & \num{22301964} & \num{14854} & \num{27.32} & \num{4} & \num{13} & \num{35} & \num{70} \\
SN6 & soc-livejournal & Social Network & \num{4033137} & \num{27933062} & \num{2651} & \num{13.85} & \num{2} & \num{5} & \num{15} & \num{35} \\
SN7 & soc-FourSquare & Social Network & \num{639014} & \num{3214986} & \num{106218} & \num{10.06} & \num{1} & \num{1} & \num{4} & \num{19} \\
TM1 & vsp\_vibrobox\_scagr7-2c\_rlfddd & Triangle Mixture & \num{77328} & \num{435586} & \num{669} & \num{11.27} & \num{3} & \num{5} & \num{8} & \num{26} \\
TM2 & vsp\_bump2\_e18\_aa01\_model1\_crew1 & Triangle Mixture & \num{56438} & \num{300801} & \num{604} & \num{10.66} & \num{5} & \num{7} & \num{11} & \num{18} \\
TM3 & vsp\_barth5\_1Ksep\_50in\_5Kout & Triangle Mixture & \num{32212} & \num{101805} & \num{22} & \num{6.32} & \num{6} & \num{6} & \num{7} & \num{7} \\
TM4 & vsp\_model1\_crew1\_cr42\_south31 & Triangle Mixture & \num{45101} & \num{189976} & \num{17663} & \num{8.42} & \num{3} & \num{5} & \num{7} & \num{9} \\
TM5 & vsp\_p0291\_seymourl\_iiasa & Triangle Mixture & \num{10498} & \num{53868} & \num{229} & \num{10.26} & \num{3} & \num{7} & \num{11} & \num{16} \\
TM6 & vsp\_bcsstk30\_500sep\_10in\_1Kout & Triangle Mixture & \num{58348} & \num{2016578} & \num{219} & \num{69.12} & \num{45} & \num{64} & \num{92} & \num{106} \\
TM7 & vsp\_befref\_fxm\_2\_4\_air02 & Triangle Mixture & \num{14109} & \num{98224} & \num{1531} & \num{13.92} & \num{6} & \num{9} & \num{11} & \num{15} \\
US1 & mi2010 & US Census Redistricting & \num{329885} & \num{789045} & \num{58} & \num{4.78} & \num{3} & \num{4} & \num{5} & \num{8} \\
WB1 & web-it-2004 & Web Graph & \num{509338} & \num{7178413} & \num{469} & \num{28.19} & \num{14} & \num{16} & \num{19} & \num{39} \\
WB2 & web-google & Web Graph & \num{1299} & \num{2773} & \num{59} & \num{4.27} & \num{1} & \num{2} & \num{5} & \num{12} \\
WB3 & web-wikipedia2009 & Web Graph & \num{1864433} & \num{4507315} & \num{2624} & \num{4.84} & \num{1} & \num{2} & \num{5} & \num{10} \\
\bottomrule
\end{tabular}
\end{table*}

 \section{Additional Figures}\label{app:figures} %
In this section we include a number of figures omitted from the main section of the paper.

\begin{figure}\Description{A scatterplot showing different graph instances in different colors, with the x axis being the normalized cut value and the y axis being the running time. Different shapes of the points denote either the dynamic programming heuristic or the greedy heuristic. The dynamic programming heuristic is usually further up along the y axis, indicating it is slower than greedy, while the normalized cut value is similar or slightly worse.}
    \centering
    \includegraphics[width=0.9\linewidth]{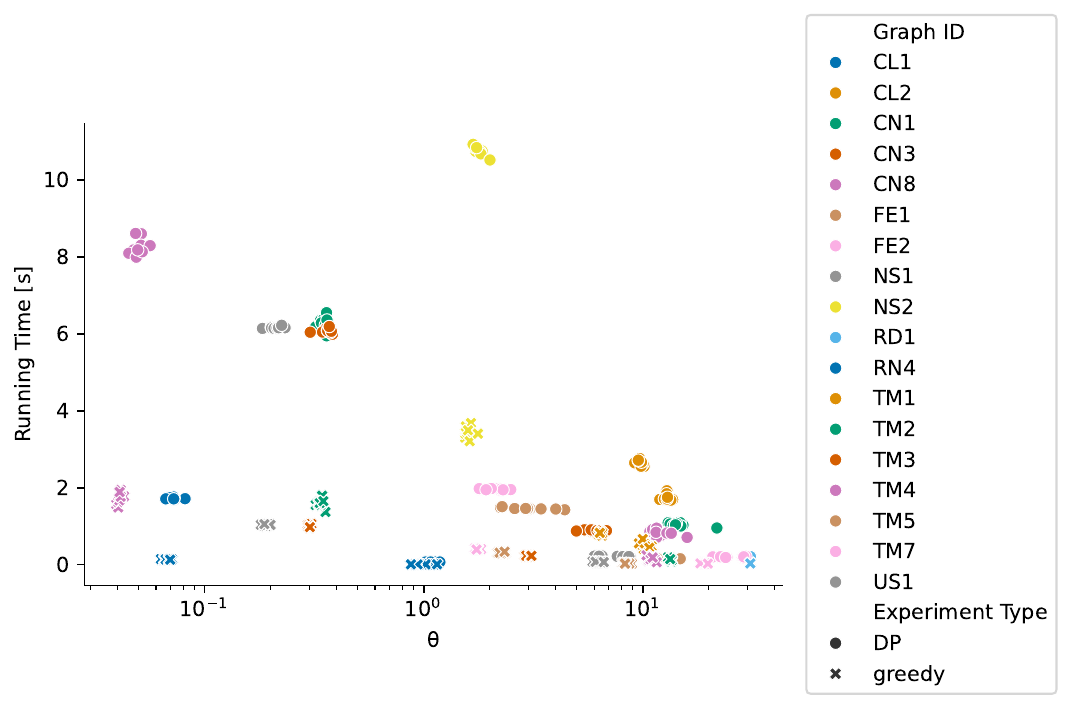}
    \caption{\mathversion{bold}\Algo{} with the greedy vs.\ dynamic programming strategy for a subset of instances for $k=32$ and $\rho = 0.0001$.}
    \label{fig:dp-comparison}
\end{figure}

\begin{figure*}[ht]
\Description{The figure depicts two bar plots stacked on top of each other. On the x axis there are the IDs of different graph instances. Each instance then has four corresponding bars for the algorithms XCut, XCut minimum, KaHiP and Metis. The y axis shows the relative change in normalized cut value for k = 2 of these algorithms versus Graclus. 
The top graph contains the IMDB graph, citation networks, social networks, email networks, infrastructure networks and web graphs. The lines of XCut point down from the 0 percent change line, indicating their values are smaller, often by more than 90 percent. METIS and KaHiP's lines always point up, reaching 25-75 percent worse values.
In the bottom figure are the remaining instances, where the divergence of most solvers is within plus/minus fifty percent and no clear pattern is visible. Some values of XCut are considerably lower than Graclus in the bottom plot.}
    \centering
    \includegraphics[width=\textwidth]{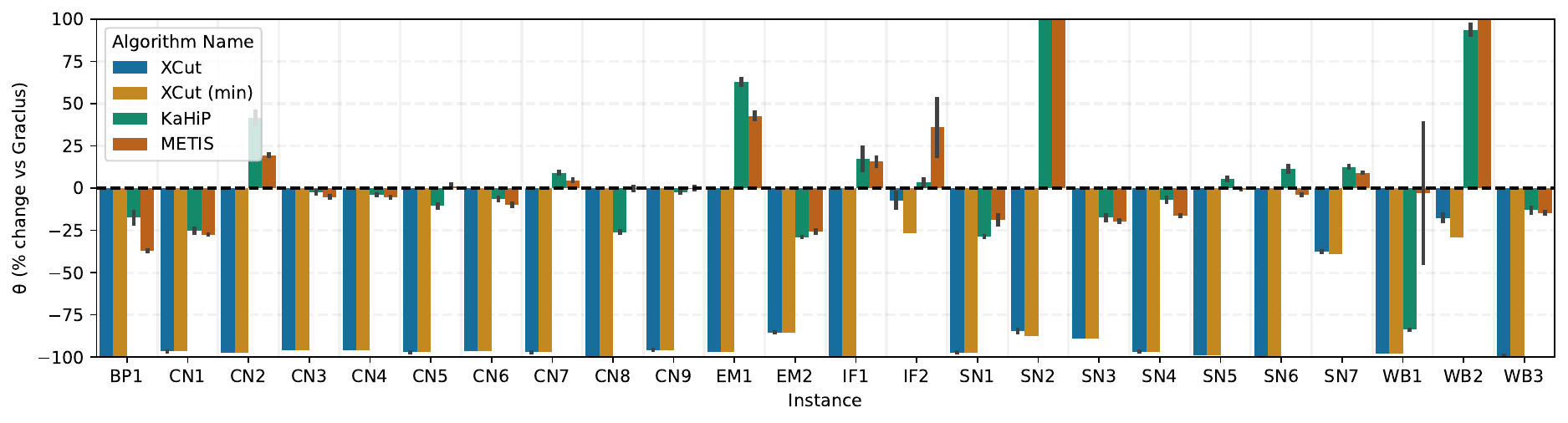}
    \\
    \includegraphics[width=\textwidth]{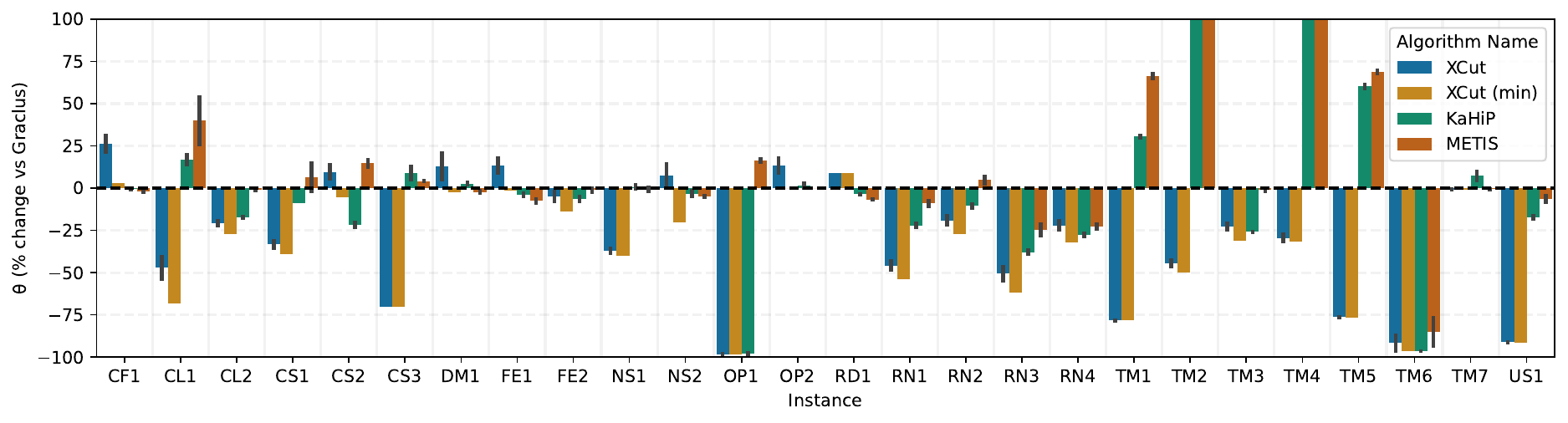}
    \caption{\mathversion{bold}Percentage deviation of the returned normalized cut value relative to \Graclus{} for $k=2$. 
    This means that a value of -75\% means that the normalized cut value is 75\% lower (i.e., better).
    The thin black bars indicate the standard error across our runs.
    The top graph shows the disconnected IMDB graph (BP1), citation network instances (CN), email networks (EM), infrastructure graphs (IF), social networks (SN) and web graphs (WB), while the bottom shows the remaining instances. See \autoref{tab:instances} for details.
    }\label{fig:barplot-2}
\end{figure*}

\begin{figure*}[ht]\Description{The figure depicts two bar plots stacked on top of each other. On the x axis there are the IDs of different graph instances. Each instance then has four corresponding bars for the algorithms XCut, XCut minimum, KaHiP and Metis. The y axis shows the relative change in normalized cut value for k = 2 of these algorithms versus Graclus. 
The top graph contains the IMDB graph, citation networks, social networks, email networks, infrastructure networks and web graphs. The lines of XCut point down from the 0 percent change line, indicating their values are smaller, often by more than 90 percent. METIS and KaHiP's lines always point up, reaching 25-75 percent worse values.
In the bottom figure are the remaining instances, where the divergence of most solvers is within plus/minus fifty percent and no clear pattern is visible.}
    \centering
    \includegraphics[width=\textwidth]{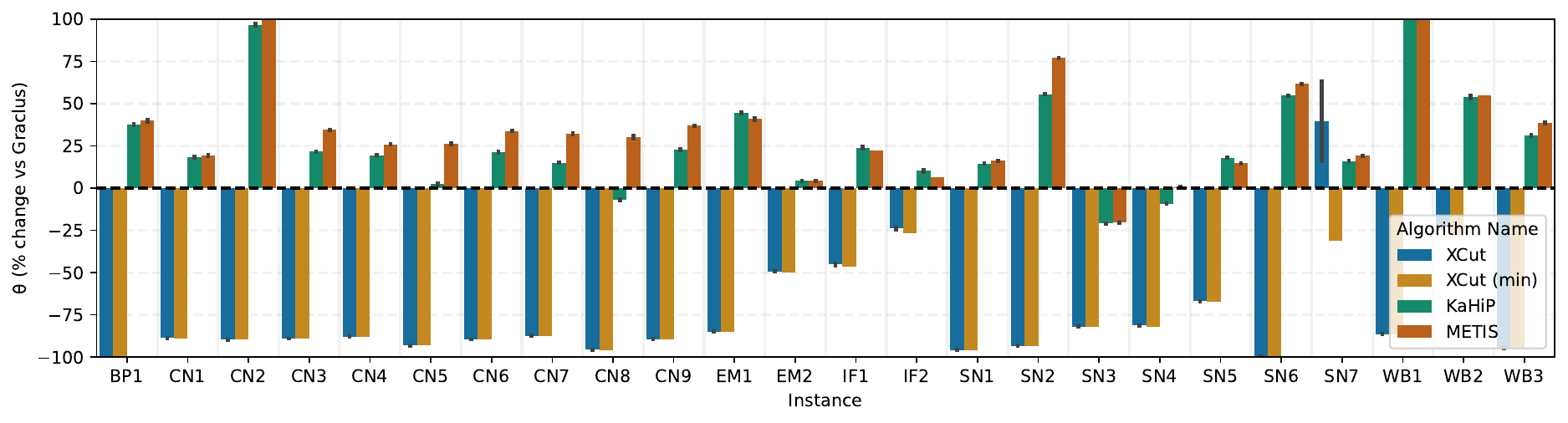}
    \\
    \includegraphics[width=\textwidth]{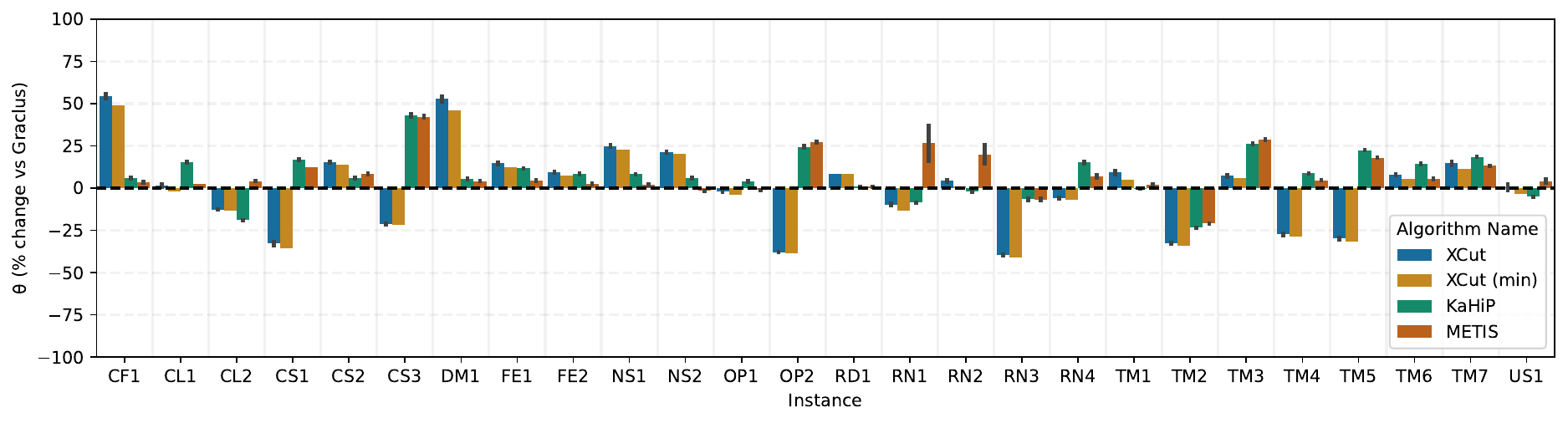}
    \caption{\mathversion{bold}Percentage deviation of the returned normalized cut value relative to \Graclus{} for $k=128$. 
    This means that a value of -75\% means that the normalized cut value is 75\% lower (i.e., better).
    The thin black bars indicate the standard error across our runs.
    The top graph shows the disconnected IMDB graph (BP1), citation network instances (CN), email networks (EM), infrastructure graphs (IF), social networks (SN) and web graphs (WB), while the bottom shows the remaining instances. See \autoref{tab:instances} for details.
    }
    \label{fig:barplot-128}
\end{figure*}

\begin{figure*}\Description{A series of five plots that depict the difference in running time for XCut on different graph instances. On the x axis is k, on the y axis is the running time in seconds. In orange is the time it takes to compute the expander decomposition, in blue is the total time to compute a normalized cut. The blue bars increase in size as k increases. On the smaller graph, the blue bar makes up for up to 60 percent of the running time for k = 128, and as the graphs get larger, the proportion shrinks, with it making up at most 10 percent for the largest graph, SN5.}
    \centering
    \includegraphics[width=\linewidth]{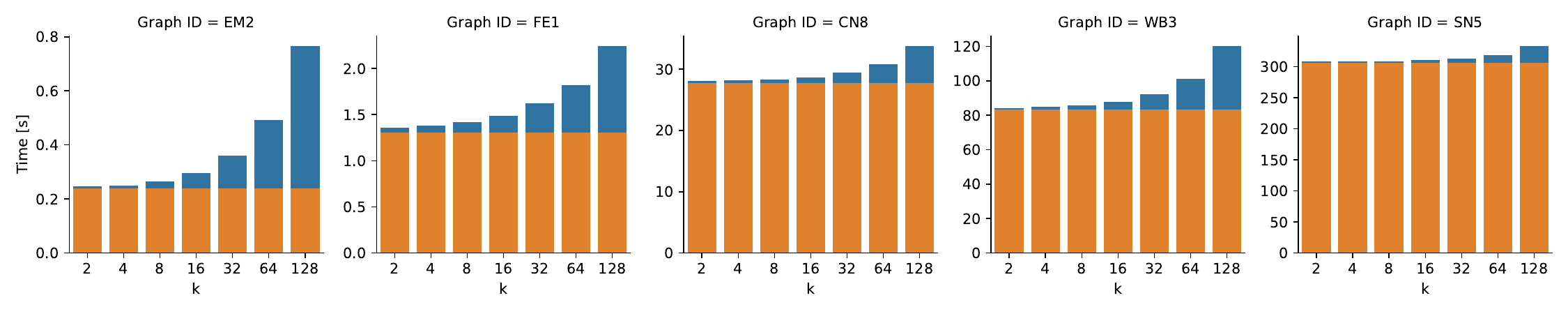}
    \caption{\mathversion{bold}Plot showcasing the time taken to compute the expander hierarchy in orange and the total time to compute a normalized cut in blue for five graph instances of different sizes.}
    \label{fig:stacked-bars-2}
\end{figure*}

\begin{figure*}\Description{A grid of 10 plots, each depicting the running time of XCut, Metis, KaHiP and Graclus as k increases. The top row shows XCut increasing linearly on the top, while the bottom shows different solvers behaving in an unusual manner, often increasing rapidly as k grows and taking much longer than the competition.}
    \centering
    \includegraphics[width=\linewidth]{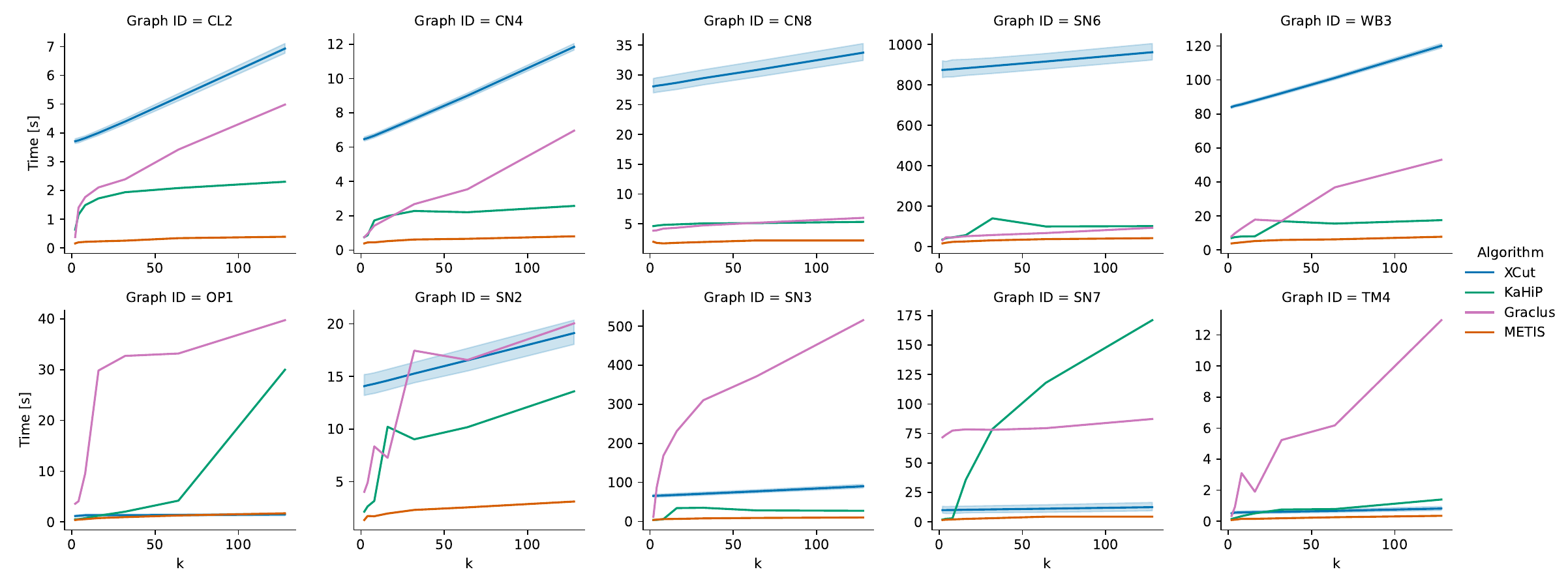}
    \caption{\mathversion{bold}Plots of running time versus $k$. The first row shows the usual behavior of the solvers in our comparison, with our solver starting at a much higher level than the competition.
    This is due to the fact that we need to first compute the expander hierarchy, which gives a hard lower bound on the running time of our algorithm for any choice of $k$. The running time then grows linearly with $k$, due to Greedy taking time $O(nk)$. The bottom row contains some instances on which the running time results are unusual. Often \Graclus{} running time grows very quickly, but we are unsure what causes this on these instances in particular.
    }\label{fig:time-examples}
\end{figure*}

\begin{table*}
\caption{\mathversion{bold}Geometric mean of normalized cut value across graph types over $k$ for the multilevel graph partitioning algorithms from \autoref{sec:graclus}.}\label{tab:avgs}
\scalebox{0.75}{
\begin{tabular}{llrrrrrrr}
\toprule
 & $k$ & 2 & 4 & 8 & 16 & 32 & 64 & 128 \\
Type & Algorithm Name &  &  &  &  &  &  &  \\
\midrule
\multirow[t]{4}{*}{Bipartite} & Graclus & 0.05 & 0.17 & 0.45 & 1.09 & 3.55 & 9.18 & 24.41 \\
 & KaHiP & 0.04 & 0.27 & 0.67 & 1.92 & 4.84 & 13.03 & 33.65 \\
 & METIS & 0.03 & 0.17 & 0.59 & 1.66 & 4.66 & 13.07 & 34.55 \\
 & XCut & \bf 0.00 & \bf 0.00 & \bf 0.00 & \bf 0.00 & \bf 0.00 & \bf 0.00 & \bf 0.00 \\
\cline{1-9}
\multirow[t]{4}{*}{Circuit Simulation} & Graclus & \bf 0.01 & 0.06 & \bf 0.18 & \bf 0.54 & 1.57 & 4.76 & 14.04 \\
 & KaHiP & \bf 0.01 & 0.07 & 0.21 & 0.62 & 1.78 & 5.65 & 16.99 \\
 & METIS & 0.02 & 0.07 & 0.22 & 0.59 & 1.72 & 5.49 & 16.90 \\
 & XCut & \bf 0.01 & \bf 0.05 & \bf 0.18 & \bf 0.54 & \bf 1.41 & \bf 4.21 & \bf 11.90 \\
\cline{1-9}
\multirow[t]{4}{*}{Citation Network} & Graclus & 0.08 & 0.26 & 0.67 & 1.57 & 3.75 & 8.22 & 17.76 \\
 & KaHiP & 0.07 & 0.27 & 0.73 & 1.81 & 4.27 & 9.57 & 21.47 \\
 & METIS & 0.08 & 0.29 & 0.81 & 2.01 & 4.80 & 11.04 & 24.32 \\
 & XCut & \bf 0.00 & \bf 0.01 & \bf 0.03 & \bf 0.08 & \bf 0.22 & \bf 0.61 & \bf 1.69 \\
\cline{1-9}
\multirow[t]{4}{*}{Clustering} & Graclus & 0.04 & 0.14 & 0.43 & 1.20 & 3.14 & 8.19 & 20.69 \\
 & KaHiP & 0.04 & 0.13 & 0.43 & 1.15 & 3.11 & 7.91 & 19.95 \\
 & METIS & 0.04 & 0.15 & 0.46 & 1.25 & 3.37 & 8.58 & 21.38 \\
 & XCut & \bf 0.02 & \bf 0.09 & \bf 0.29 & \bf 0.93 & \bf 2.52 & \bf 7.07 & \bf 19.49 \\
\cline{1-9}
\multirow[t]{4}{*}{Computational Fluids} & Graclus & \bf 0.12 & 0.46 & \bf 1.32 & \bf 3.86 & \bf 10.49 & \bf 27.76 & \bf 70.28 \\
 & KaHiP & \bf 0.12 & 0.46 & 1.47 & 4.18 & 10.98 & 29.64 & 74.51 \\
 & METIS & \bf 0.12 & \bf 0.44 & 1.35 & 4.10 & 10.95 & 28.65 & 72.55 \\
 & XCut & 0.15 & 0.55 & 1.63 & 4.75 & 14.30 & 44.39 & 108.44 \\
\cline{1-9}
\multirow[t]{4}{*}{Duplicate Materials} & Graclus & \bf 0.04 & \bf 0.24 & \bf 0.80 & \bf 2.58 & \bf 7.71 & \bf 20.42 & \bf 53.60 \\
 & KaHiP & 0.05 & 0.32 & 0.85 & 2.97 & 8.20 & 21.45 & 56.65 \\
 & METIS & \bf 0.04 & \bf 0.24 & \bf 0.80 & 2.60 & 7.75 & 20.67 & 55.85 \\
 & XCut & 0.05 & 0.26 & 0.90 & 3.27 & 9.54 & 26.81 & 81.90 \\
\cline{1-9}
\multirow[t]{4}{*}{Email Network} & Graclus & 0.18 & 0.60 & 1.65 & 4.07 & 9.11 & 20.13 & 42.43 \\
 & KaHiP & 0.19 & 0.66 & 1.70 & 4.36 & 10.11 & 23.19 & 51.73 \\
 & METIS & 0.19 & 0.64 & 1.80 & 4.40 & 10.21 & 23.06 & 51.24 \\
 & XCut & \bf 0.01 & \bf 0.04 & \bf 0.17 & \bf 0.51 & \bf 1.46 & \bf 4.07 & \bf 11.76 \\
\cline{1-9}
\multirow[t]{4}{*}{Finite Elements} & Graclus & \bf 0.01 & \bf 0.06 & 0.21 & \bf 0.63 & \bf 1.82 & \bf 5.16 & \bf 13.97 \\
 & KaHiP & \bf 0.01 & \bf 0.06 & 0.21 & 0.66 & 2.00 & 5.63 & 15.40 \\
 & METIS & \bf 0.01 & \bf 0.06 & \bf 0.20 & 0.65 & 1.88 & 5.23 & 14.42 \\
 & XCut & \bf 0.01 & \bf 0.06 & 0.21 & 0.67 & 1.99 & 5.67 & 15.64 \\
\cline{1-9}
\multirow[t]{4}{*}{Infrastructure Network} & Graclus & 0.02 & 0.09 & 0.39 & 1.09 & 3.19 & 9.97 & 27.44 \\
 & KaHiP & 0.02 & 0.09 & 0.44 & 1.51 & 4.41 & 12.24 & 32.27 \\
 & METIS & 0.02 & 0.11 & 0.49 & 1.36 & 4.16 & 11.60 & 31.29 \\
 & XCut & \bf 0.00 & \bf 0.00 & \bf 0.00 & \bf 0.23 & \bf 1.38 & \bf 5.30 & \bf 17.73 \\
\cline{1-9}
\multirow[t]{4}{*}{Numerical Simulation} & Graclus & 0.02 & \bf 0.08 & 0.28 & 0.94 & 2.70 & \bf 7.24 & \bf 19.40 \\
 & KaHiP & 0.02 & 0.09 & 0.31 & 0.93 & 2.76 & 7.77 & 20.76 \\
 & METIS & 0.02 & 0.09 & \bf 0.26 & \bf 0.89 & \bf 2.63 & 7.25 & 19.43 \\
 & XCut & \bf 0.01 & 0.09 & 0.33 & 1.10 & 3.18 & 8.88 & 23.89 \\
\cline{1-9}
\multirow[t]{4}{*}{Optimization} & Graclus & 0.04 & 0.14 & 0.40 & 1.85 & 5.65 & 16.68 & 42.32 \\
 & KaHiP & \bf 0.00 & 0.07 & 0.32 & 1.86 & 5.84 & 18.65 & 48.48 \\
 & METIS & 0.04 & 0.09 & 0.40 & 1.83 & 5.55 & 16.60 & 47.46 \\
 & XCut & 0.01 & \bf 0.03 & \bf 0.22 & \bf 1.38 & \bf 4.63 & \bf 13.60 & \bf 33.06 \\
\cline{1-9}
\multirow[t]{4}{*}{Random Graph} & Graclus & 0.91 & \bf 2.59 & \bf 6.19 & \bf 13.52 & 28.37 & \bf 57.79 & \bf 116.97 \\
 & KaHiP & 0.88 & 2.67 & 6.37 & 13.83 & 28.62 & 58.17 & 118.02 \\
 & METIS & \bf 0.85 & 2.61 & \bf 6.19 & 13.53 & \bf 28.36 & 58.18 & 117.98 \\
 & XCut & 0.99 & 2.99 & 6.98 & 14.99 & 30.99 & 62.99 & 126.99 \\
\cline{1-9}
\multirow[t]{4}{*}{Road Network} & Graclus & \bf 0.00 & \bf 0.00 & \bf 0.00 & \bf 0.01 & \bf 0.03 & 0.10 & \bf 0.30 \\
 & KaHiP & \bf 0.00 & \bf 0.00 & \bf 0.00 & \bf 0.01 & \bf 0.03 & \bf 0.09 & \bf 0.30 \\
 & METIS & \bf 0.00 & \bf 0.00 & \bf 0.00 & \bf 0.01 & \bf 0.03 & 0.10 & \bf 0.30 \\
 & XCut & \bf 0.00 & \bf 0.00 & \bf 0.00 & \bf 0.01 & \bf 0.03 & 0.11 & 0.34 \\
\cline{1-9}
\multirow[t]{4}{*}{Social Network} & Graclus & 0.16 & 0.66 & 1.64 & 4.01 & 9.03 & 19.96 & 44.26 \\
 & KaHiP & 0.20 & 0.73 & 1.92 & 4.66 & 10.59 & 24.11 & 51.14 \\
 & METIS & 0.19 & 0.68 & 1.87 & 4.67 & 10.91 & 24.54 & 53.15 \\
 & XCut & \bf 0.01 & \bf 0.03 & \bf 0.11 & \bf 0.31 & \bf 0.83 & \bf 2.19 & \bf 6.02 \\
\cline{1-9}
\multirow[t]{4}{*}{Triangle Mixture} & Graclus & 0.08 & 0.43 & 1.43 & 3.64 & 8.94 & 20.71 & 48.90 \\
 & KaHiP & 0.08 & 0.50 & 1.54 & 3.94 & 9.62 & 22.73 & 52.85 \\
 & METIS & 0.09 & 0.56 & 1.60 & 4.14 & 10.13 & 23.10 & 51.51 \\
 & XCut & \bf 0.03 & \bf 0.25 & \bf 1.08 & \bf 3.08 & \bf 7.70 & \bf 18.44 & \bf 44.26 \\
\cline{1-9}
\multirow[t]{4}{*}{US Census Redistricting} & Graclus & \bf 0.00 & 0.01 & \bf 0.02 & 0.07 & 0.21 & 0.64 & 1.92 \\
 & KaHiP & \bf 0.00 & \bf 0.00 & \bf 0.02 & \bf 0.06 & 0.20 & \bf 0.59 & \bf 1.80 \\
 & METIS & \bf 0.00 & 0.01 & \bf 0.02 & 0.07 & 0.21 & 0.68 & 1.99 \\
 & XCut & \bf 0.00 & \bf 0.00 & \bf 0.02 & 0.07 & \bf 0.19 & \bf 0.59 & 1.92 \\
\cline{1-9}
\multirow[t]{4}{*}{Web Graph} & Graclus & \bf 0.00 & 0.02 & 0.04 & 0.10 & 0.26 & 0.77 & 2.02 \\
 & KaHiP & \bf 0.00 & 0.01 & 0.04 & 0.11 & 0.47 & 2.09 & 7.36 \\
 & METIS & 0.01 & 0.02 & 0.03 & 0.08 & 0.42 & 2.46 & 9.01 \\
 & XCut & \bf 0.00 & \bf 0.00 & \bf 0.00 & \bf 0.01 & \bf 0.02 & \bf 0.10 & \bf 0.36 \\
\bottomrule
\end{tabular}

}
\end{table*}

\begin{table*}%
\caption{\mathversion{bold}Cut values and running time of our Algorithm vs the values reported by Zhao et al~\cite{zhao2018nearly}. All values are for $k=30$. The Cut-columns contain the normalized cut value (lower is better). The minimum value in each row is marked bold. The bottom row contains the geometric mean of all values.
}\label{tab:zhao}
\setlength{\tabcolsep}{2pt}
\begin{tabular}{lrrrr}
\toprule
GID & \Algo{mean} & \Algo{min} & Zhao & \Graclus \\
\midrule
CL1 & 0.87 & \bf 0.77 & 1.05 & 1.15 \\
CL2 & 6.00 & \bf 5.87 & 7.05 & 7.61 \\
CN1 & 0.31 & \bf 0.29 & 0.52 & 4.06 \\
CN3 & \bf 0.27 & \bf 0.27 & 0.49 & 3.66 \\
CN5 & \bf 0.13 & \bf 0.13 & 0.14 & 3.24 \\
CN7 & 0.17 & \bf 0.16 & 0.41 & 2.21 \\
CN8 & \bf 0.04 & \bf 0.04 & 0.06 & 1.88 \\
FE1 & 2.06 & 1.96 & \bf 1.68 & 1.74 \\
FE2 & 1.58 & 1.55 & 1.50 & \bf 1.45 \\
NS1 & 5.73 & 5.35 & \bf 4.71 & \bf 4.71 \\
NS2 & 1.45 & 1.39 & \bf 1.08 &  1.17 \\
RD1 & 28.98 & 28.98 & \bf 23.80 & 26.48 \\
RN4 & \bf 0.06 & \bf 0.06 & 0.07 & 0.07 \\
TM1 & 9.53 & 8.91 & \bf 6.85 & 8.48 \\
TM2 & 12.34 & \bf 12.06 & 13.55 & 16.63 \\
TM3 & 2.78 & \bf 2.65 & 2.72 & 2.78 \\
TM4 & 10.09 & \bf 9.58 & 10.48 & 14.45 \\
TM5 & 7.77 & \bf 7.48 & 7.88 & 13.25 \\
TM6 & 2.11 & \bf 1.94 & 2.09 & 2.1 \\
TM7 & 17.51 & 16.90 & \bf 12.83 & 15.6 \\
US1 & 0.17 & \bf 0.16 & 0.41 & 0.19 \\
\midrule
All & 1.46 & \bf 1.39 & 1.64 & 3.06 \\
\bottomrule
\end{tabular}
\end{table*}

\end{document}